\documentclass[letterpaper,11pt]{article}
\usepackage[utf8]{inputenc}

\pdfoutput=1

\usepackage{amsmath, amsthm, amssymb}
\usepackage[linesnumbered,vlined]{algorithm2e}
\usepackage{mathtools}
\usepackage[numbers]{natbib}
\usepackage{comment} 
\usepackage{color-edits} 
\usepackage[most]{tcolorbox}
\usepackage{xfrac}
\usepackage{hyperref}
\usepackage{multirow}
\usepackage{caption}
\usepackage{bm}
\usepackage{newfloat}
\usepackage{enumitem}
\usepackage{xpatch}
\usepackage{soul}

\usepackage[margin=1in]{geometry}

\allowdisplaybreaks

\definecolor{mygreen}{RGB}{20,120,60}

\title{Fully Dynamic Maximal Independent Set with\\ Polylogarithmic Update Time\footnote{A preliminary version of this paper is to appear in the proceedings of {\em The 60th Annual IEEE Symposium on Foundations of Computer Science} (FOCS 2019).\smallskip\newline Portions of the work were completed while some of the authors visited The Simons Institute for Theory of Computing.\smallskip\smallskip}}

\date{}

\author{Soheil Behnezhad\thanks{\texttt{\{soheil, mahsa, hajiagha\}@cs.umd.edu}. Supported in part by NSF CAREER award CCF-1053605, NSF AF:Medium grant CCF-1161365, NSF BIGDATA grant IIS-1546108, and NSF SPX grant CCF-1822738. Soheil Behnezhad was also supported in part by a Google PhD Fellowship.}\\University of Maryland
\and Mahsa Derakhshan\footnotemark[2]\\University of Maryland
\and MohammadTaghi Hajiaghayi\footnotemark[2]\\University of Maryland
\and Cliff Stein\thanks{{\tt cliff@ieor.columbia.edu}.
Research partly supported by
NSF Grants  CCF-1714818 and CCF-1822809.}\\Columbia University
\and Madhu Sudan\thanks{{\tt madhu@cs.harvard.edu}. Supported in part by a Simons Investigator Award and NSF Award CCF 1715187.}\\Harvard University
}

\newcommand{\E}[0]{\ensuremath{\mathbb{E}}}

\newcommand{\A}[0]{\ensuremath{\mathcal{A}}}

\newcommand{\affectedset}[0]{\ensuremath{\mathcal{A}}}
\newcommand{\flippedset}[0]{\ensuremath{\mathcal{F}}}
\renewcommand{\P}[1]{\ensuremath{\mathcal{P}(#1)}}
\renewcommand{\S}[0]{\ensuremath{\mathcal{S}}}
\DeclareMathOperator*{\argmin}{arg\,min}

\DeclareMathOperator{\poly}{poly}
\DeclareMathOperator{\polylog}{polylog}

\renewcommand{\O}[1]{\ensuremath{O(#1)}}

\DeclareMathOperator{\elim}{elim}

\let\originalleft\left
\let\originalright\right
\renewcommand{\left}{\mathopen{}\mathclose\bgroup\originalleft}
\renewcommand{\right}{\aftergroup\egroup\originalright}

\addauthor{sb}{blue}    
\addauthor{md}{red}    

\newtheorem{theorem}{Theorem}

\newtheorem{lemma}{Lemma}[section]

\newtheorem{proposition}[lemma]{Proposition}
\newtheorem{corollary}[lemma]{Corollary}

\newtheorem{definition}[lemma]{Definition}
\newtheorem{claim}[lemma]{Claim}

\newtheorem{observation}[lemma]{Observation}

\newtheorem{invariant}[lemma]{Invariant}

\makeatletter
\def\thm@space@setup{%
  \thm@preskip= 0.2cm
  \thm@postskip=\thm@preskip 
}
\makeatother

\definecolor{mygreen}{RGB}{20,125,20}
\definecolor{myred}{RGB}{125,20,20}
\definecolor{linkcolor}{RGB}{0,0,230}
\definecolor{mylightgray}{RGB}{230,230,230}
\definecolor{verylightgray}{RGB}{240,240,240}
\definecolor{commentcolor}{RGB}{120,120,120}

\SetCommentSty{mycommfont}

\newcommand{\smparagraph}[1]{
\par\addvspace{0.2cm}
\noindent \textbf{#1}
}

\hypersetup{
     colorlinks=true,
     citecolor= mygreen,
     linkcolor= myred,
     urlcolor= mygreen
}

\newcommand{\etal}[0]{\textit{et al.}}

\newcounter{myalgctr}

\newenvironment{tbox}{
\par\addvspace{0.2cm}
\begin{tcolorbox}[width=\textwidth,
                  enhanced,
                  boxsep=2pt,
                  left=1pt,
                  right=1pt,
                  top=4pt,
                  boxrule=1pt,
                  arc=0pt,
                  colback=white,
                  colframe=black,
                  unbreakable
                  ]
}{
\end{tcolorbox}
}

\newenvironment{graytbox}{
\par\addvspace{0.1cm}
\begin{tcolorbox}[width=\textwidth,
                  enhanced,
                  frame hidden,
                  boxsep=5pt,
                  left=1pt,
                  right=1pt,
                  top=2pt,
                  bottom=2pt,
                  boxrule=1pt,
                  arc=0pt,
                  colback=mylightgray,
                  colframe=black,
                  breakable
                  ]
}{
\end{tcolorbox}
}

\newcommand{\tboxhrule}[0]{\vspace{0.1cm} \hrule \vspace{0.2cm}}

\newenvironment{titledtbox}[1]{\begin{tbox}#1 \tboxhrule}{\end{tbox}}

\newenvironment{tboxalg2e}[1]{
\refstepcounter{myalgctr}
	\begin{titledtbox}{\textbf{Algorithm \themyalgctr.} #1}
	\vspace{-0.2cm}
}
{
	\vspace{-0.3cm}
	\end{titledtbox}
}

\newenvironment{highlighttechnical}[0]{
\vspace{0.1cm}
\begin{tcolorbox}[width=\textwidth,
                  enhanced,
                  boxsep=2pt,
                  left=1pt,
                  right=1pt,
                  top=4pt,
                  boxrule=0.8pt,
                  arc=0pt,
                  colback=verylightgray,
                  colframe=black,
                  unbreakable
                  ]
}{
\end{tcolorbox}
}

\newenvironment{highlighttechnicalwhite}[0]{
\vspace{0.1cm}
\begin{tcolorbox}[width=\textwidth,
                  enhanced,
                  boxsep=2pt,
                  left=1pt,
                  right=1pt,
                  top=4pt,
                  boxrule=0.8pt,
                  arc=0pt,
                  colback=white,
                  colframe=black,
                  unbreakable
                  ]
}{
\end{tcolorbox}
}

\newcommand{\restatedesc}[3]{
\par\addvspace{0.2cm}
\noindent\textbf{#1 {\normalfont (#2)}.} {\em #3}
\par\addvspace{0.2cm}
}

\newcommand{\restate}[2]{
\restatedesc{#1}{restated}{#2}
}

\begin{document}
\maketitle

\thispagestyle{empty}
\begin{abstract}
\setlength{\parskip}{0.8em}
We present the first algorithm for maintaining a {\em maximal independent set} (MIS) of a {\em fully dynamic} graph---which undergoes both edge insertions and deletions---in polylogarithmic time. Our algorithm is randomized and, per update, takes $O(\log^2 \Delta \cdot \log^2 n)$ expected time. Furthermore, the algorithm can be adjusted to have $O(\log^2 \Delta \cdot \log^4 n)$ worst-case update-time with high probability. Here, $n$ denotes the number of vertices and $\Delta$ is the maximum degree in the graph.

The MIS problem in fully dynamic graphs has attracted significant attention after a breakthrough result of Assadi, Onak, Schieber, and Solomon \mbox{[STOC'18]} who presented an algorithm with $O(m^{3/4})$ update-time (and thus broke the natural $\Omega(m)$ barrier) where $m$ denotes the number of edges in the graph. This result was improved in a series of subsequent papers, though, the update-time remained polynomial. In particular, the fastest algorithm prior to our work had $\widetilde{O}(\min\{\sqrt{n}, m^{1/3}\})$ update-time [Assadi~\etal{} SODA'19].

Our algorithm maintains the {\em lexicographically first MIS} over a random order of the vertices. As a result, the same algorithm also maintains a 3-approximation of {\em correlation clustering}. We also show that a simpler variant of our algorithm can be used to maintain a random-order {\em lexicographically first maximal matching} in the same update-time.
\end{abstract}
\clearpage

\thispagestyle{empty}
{
\hypersetup{linkcolor=black}
\tableofcontents{}
}
\clearpage

\setcounter{page}{1}

\section{Introduction}\label{sec:intro}

A maximal independent set (MIS) of a graph is a fundamental object with countless theoretical and practical applications. It is one of the most well-studied problems in distributed and parallel settings following the seminal works of \cite{DBLP:conf/stoc/Luby85, DBLP:journals/jal/AlonBI86}. MIS has also been studied in a variety of other models and has diverse applications such as approximating matching and vertex cover~\cite{DBLP:conf/focs/NguyenO08,DBLP:conf/stoc/YoshidaYI09}, graph coloring~\cite{DBLP:conf/stoc/Luby85,DBLP:conf/focs/Linial87}, clustering~\cite{correlationclustering}, leader-election~\cite{DBLP:conf/podc/DaumGKN12}, and many others.

In this paper, we consider MIS in {\em fully-dynamic} graphs. The graph is updated via both edge insertions and deletions and the goal is to maintain an MIS by the end of each update. Dynamic graphs constitute an active area of research and have seen a plethora of results over the past two decades. The MIS problem in dynamic graphs has also attracted a significant attention, especially recently \cite{Censor-HillelHK16,assadistoc, DBLP:journals/corr/abs-1804-01823, DBLP:journals/corr/abs-1804-08908,onakicalp,assadisoda}. We overview these works below.



\smparagraph{Related Work on Dynamic MIS.} In static graphs with $m$ edges, a simple greedy algorithm can find an MIS in $O(m)$ time. As such, one can trivially maintain MIS by recomputing it from scratch after each update, in $O(m)$ time. In a pioneering work, Censor-Hillel, Haramaty, and Karnin~\cite{Censor-HillelHK16} presented a round-efficient randomized algorithm for MIS in dynamic {\em distributed} networks. Implementing the algorithm of \cite{Censor-HillelHK16} in the sequential setting---the focus of this paper---requires $\Omega(\Delta)$ update-time (see \cite[Section 6]{Censor-HillelHK16}) where $\Delta$ is the maximum-degree in the graph which can be as large as $\Omega(n)$ or even $\Omega(m)$ for sparse graphs. Improving this bound was one of the major problems the authors left open. Later, in a breakthrough, Assadi, Onak, Schieber, and Solomon \cite{assadistoc} presented a deterministic algorithm with $O(m^{3/4})$ update-time; thereby improving the $O(m)$ bound for all graphs. This result was further improved in a series of subsequent papers \cite{DBLP:journals/corr/abs-1804-01823, DBLP:journals/corr/abs-1804-08908,onakicalp,assadisoda}. The current state-of-the-art is a randomized algorithm due to Assadi~\etal{}~\cite{assadisoda}, which requires $\widetilde{O}(\min\{\sqrt{n}, m^{1/3}\})$ amortized update-time in $n$-vertex graphs.

\smparagraph{Our Contribution.} In this paper, we show that it is possible to maintain an MIS of fully-dynamic graphs in polylogarithmic time. This exponentially improves over the prior algorithms, which all have polynomial update-time on general graphs. Our algorithm is randomized and requires the standard {\em oblivious adversary}\footnote{In the standard {\em oblivious adversarial} model, the adversary can feed in any sequence of edge updates and is aware of the algorithm to be used, but is unaware of the random-bits used by the algorithm. Equivalently, one can assume that the sequence of edge updates is picked adversarially {\em before} the dynamic algorithm starts to operate.} assumption (as do all previous randomized algorithms).

\newcommand{\thmmain}[0]{There is a data structure to maintain an MIS against an oblivious adversary in a fully-dynamic graph that, per update, takes $O(\log^2 \Delta \cdot \log^2 n)$ expected time. Furthermore, the number of adjustments to the MIS per update is $O(1)$ in expectation.
}
\begin{graytbox}
	\begin{theorem}[main result]\label{thm:main}
		\thmmain{}	
	\end{theorem}
\end{graytbox}

Since our algorithm bounds the expected time  {\em per update} without amortization, we can use it as a black-box in a framework of Bernstein~\etal{}~\cite[Theorem 1.1]{DBLP:conf/soda/BernsteinFH19} to also get a worst-case  guarantee w.h.p.\footnote{Here and throughout the paper, ``w.h.p.'' abbreviates ``with high probability'' and implies probability at least $1-n^{-c}$ for any desirably large constant $c>1$ that may affect the hidden constants in the bounds.} (We note that this comes at the cost of losing the guarantee on the adjustment-complexity.)


\begin{corollary}
	There is a data structure to maintain an MIS against an oblivious adversary in a fully-dynamic graph that w.h.p. has $O(\log^2 \Delta \cdot \log^4 n)$ worst-case update-time.
\end{corollary}

To prove Theorem~\ref{thm:main}, we give an algorithm that carefully simulates the {\em lexicographically first MIS} (LFMIS) over a random ranking of the vertices (see Section~\ref{sec:preliminaries} for definition). Once this order is fixed, the LFMIS of the graph becomes unique. This is particularly useful for dynamic graphs as it makes the output {\em history-independent}. That is, the order of edge insertions and deletions by the adversary cannot affect the reported MIS. See \cite[Section 5]{Censor-HillelHK16} for more discussion on this property and also \cite{DBLP:conf/spaa/BlellochFS12, DBLP:conf/focs/NguyenO08} for some other useful features of random-order LFMIS.

We note that maintaining LFMIS over a random permutation has been done before by Censor-Hillel~\etal{}~\cite{Censor-HillelHK16} and also partially by Assadi~\etal{}~\cite{assadisoda} who combined it with another deterministic algorithm. However, as discussed above, both these algorithms require a polynomial update-time. The novelty of our approach is in (1) {\em the algorithm and data structures} with which we maintain this MIS, and (2) the {\em analysis} of why polylogarithmic time is sufficient. The high-level intuitions behind both the algorithm and the analysis are presented in Section~\ref{sec:highlevel}.

\smparagraph{Independent Work.} Independently and concurrently, Chechik and Zhang \cite{independentwork} also came up with an algorithm for maintaining a fully-dynamic MIS in polylogarithmic update-time against an oblivious adversary. Similar to our algorithm, they also maintain a random-order LFMIS. For general graphs with arbitrary maximum degree, both our algorithm and that of \cite{independentwork} take $O(\log^4 n)$ worst-case expected update-time.

\subsection{Other Implications of our Approach}

\smparagraph{Correlation Clustering.} Due to a reduction of \mbox{Ailon \etal{}~\cite{correlationclustering}}, our algorithm with essentially no change also maintains a 3-approximation of min-disagreement {\em correlation clustering} using the same update-time. 

\begin{corollary}
	There is a data structure to maintain a 3-approximation of the min disagreement variant of correlation clustering on completely labeled graphs against an oblivious adversary in a fully-dynamic graph that per update, takes $O(\log^2 \Delta \cdot \log^2 n)$ expected time. Furthermore, the number of changes to the clusters  per update is $O(1)$ in expectation.
\end{corollary}

\smparagraph{Maximal Matching.} There has been a huge body of work on the matching problem in dynamic graphs, see e.g. \cite{DBLP:conf/stoc/OnakR10,DBLP:journals/siamcomp/BaswanaGS18,DBLP:conf/stoc/NeimanS13,DBLP:conf/focs/GuptaP13,DBLP:journals/siamcomp/BhattacharyaHI18,DBLP:conf/stoc/BhattacharyaHN16,DBLP:conf/focs/Solomon16,DBLP:conf/soda/BhattacharyaHN17, DBLP:conf/ipco/BhattacharyaCH17,DBLP:conf/stoc/GuptaK0P17,DBLP:conf/icalp/CharikarS18,DBLP:conf/icalp/ArarCCSW18,DBLP:conf/soda/BernsteinFH19} and the references therein. Among these results, maintaining a maximal matching (MM) has been of special interest. MIS and MM are closely related. Despite all the similarities, however, the known algorithms for MM were much more efficient \cite{DBLP:journals/siamcomp/BaswanaGS18,DBLP:conf/focs/Solomon16,DBLP:conf/soda/BernsteinFH19}. Assadi~\etal{}~\cite[Section 1.1]{assadistoc} in part justified this by describing why the common techniques used for maintaining MM are not applicable to MIS, hinting also that MIS ``is inherently more complicated \cite{assadistoc}'' in fully dynamic graphs. We formalize this intuition further and show that indeed a simpler variant of our MIS algorithm can also maintain a {\em lexicographically first MM} (LFMM) over a random order on the edges, with essentially the same update-time:

\newcommand{\thmMM}[0]{There is a data structure to maintain a random-order lexicographically first maximal matching against an oblivious adversary in a fully-dynamic graph that per update, takes $O(\log^2 \Delta \cdot \log^2 n)$ expected time. Furthermore, per update, the adjustment-complexity is $O(1)$ in expectation.}

\begin{theorem}\label{thm:MM}
	\thmMM{}
\end{theorem}

This also leads to the following worst-case guarantee when used as a black-box \cite[Theorem~1.1]{DBLP:conf/soda/BernsteinFH19}.

\begin{corollary}\label{cor:matching}
	There is a data structure to maintain a maximal matching against an oblivious adversary in a fully-dynamic graph that w.h.p. has $O(\log^2 \Delta \cdot \log^4 n)$ worst-case update-time.
\end{corollary}

We emphasize that if one allows amortization, then one can get much more efficient algorithms for MM due to the seminal works of  Baswana, Gupta, and Sen \cite{DBLP:journals/siamcomp/BaswanaGS18} and Solomon \cite{DBLP:conf/focs/Solomon16}. However, our approach of maintaining random-order LFMM significantly deviates from the prior works on MM in dynamic graphs. We believe this is an important feature on its own and may find further applications.
\section{Our Techniques}\label{sec:highlevel}

\newcommand{\msnote}[1]{{\color{red} Madhu: #1}}

As pointed out earlier, our main contribution is to show that it is actually possible to maintain the lexicographically first MIS (LFMIS), under a random ordering of the vertices, at an expected polylogarithmic cost per update. In this section we attempt to explain some of the barriers and how our work overcomes them.


The first hurdle behind maintaining the LFMIS is that it may change a lot under updates. But it is also well-known~\cite[Theorem~1]{Censor-HillelHK16} that for a random ordering, the expected alteration to the LFMIS after the insertion or deletion of a single edge is $O(1)$. This already shows that maintaining random order LFMIS is sufficient to get an algorithm with $O(1)$ expected adjustments per update. However, it is not clear how to {\em detect} these changes and maintain the LFMIS efficiently: The natural algorithm to do so would do a breadth-first-search (BFS) from the endpoints of the edge being updated, but even exploring the neighborhood of a single vertex of degree $\Delta$ might require $\Omega(\Delta)$ time which is prohibitively expensive for general graphs where $\Delta$ can be as large as $\Omega(n)$.\footnote{This is precisely the $\Omega(\Delta)$ barrier mentioned in Section 6 of \cite{Censor-HillelHK16}.}


Our first idea is to maintain not just the LFMIS, but also the ``eliminator'' of every vertex $v$ in the graph. Briefly, given a ranking $\pi:V \to [0,1]$, the eliminator of a vertex $v$ in a graph $G$ under $\pi$ is its neighbor in $G$ of the lowest rank that belongs to the LFMIS. (If $v$ is in the LFMIS, then its eliminator is defined to be itself.) Maintaining the eliminators only seems to complicate our task further: (1) Even if the MIS changes by a little, it is conceivable that the eliminators of many more vertices might change. (2) It is still unclear how to find the set of vertices whose eliminators have changed in $o(\Delta)$ time.

For problem (1) we extend the classical analysis \cite{Censor-HillelHK16,DBLP:conf/stoc/YoshidaYI09}, which showed that the MIS changes only by a little after each update, to show that the eliminators are also extremely robust under updates. We stress that this extension is not simple and requires many new ideas. Overall, we get the following  guarantee which may be of independent interest. (It is crucial for our analysis that we prove this bound on {\em vertex} updates---we will discuss this towards the end of this section.)

\restatedesc{Theorem~\ref{thm:vertexA}}{informal---see page \pageref{thm:vertexA} for the formal statement}{For any arbitrary vertex addition or deletion, the expected number of vertices whose eliminator changes is $O(\log n)$.}


We now turn to problem (2), i.e., the challenge of maintaining information such as membership in the MIS and eliminators of vertices. Consider an edge update $(a,b)$ with $\pi(a) < \pi(b)$ and suppose that this changes $b$'s MIS-status. A priori, this seems to require exploring every neighbor of $b$ (at the very least) and checking to see if their status or eliminator changes. But a quick examination reveals we only need to explore those neighbors $u$ of $b$ whose eliminators have rank larger than $\pi(a)$. (Vertices with rank less than $\pi(a)$ don't change their membership in the MIS, and so vertices with eliminators of rank less than $\pi(a)$ don't change their eliminator.) To help this prune our exploration space, it would make sense to store all neighbors of $b$ (and of every vertex for that matter) in a search tree indexed by the rank of their eliminator and indeed this is an idea we pursue. However maintaining every neighbor of $b$ indexed by its eliminator-rank leads to new maintenance problems: Up to $\Delta$ trees may need to be updated when $b$ changes its eliminator-rank! We overcome this barrier with the following solution (which is essentially our final solution): We only maintain the neighbors of {\em low-rank} in a search tree indexed by eliminator-ranks and maintain the neighbors of high-rank in a more static tree indexed by just their ID (i.e., their name). 

Specifically, for each vertex $v$, we partition its neighborhood (dynamically) in two parts, $N^-(v)$ and $N^+(v)$ as described next.
The set $N^-(v)$ includes neighbors of $v$ whose eliminators have smaller rank than the eliminator of $v$. Each vertex $u \in N^-(v)$ is indexed by the rank of its (dynamically changing) eliminator. The set $N^+(v)$ includes the rest of neighbors of $v$ and every vertex $u \in N^+(v)$ is indexed by its (static) ID. Armed with these data structures it turns out one can implement updates in expected time $\polylog n$ per {\em affected vertex}, i.e., those whose eliminator has changed. (See Lemma~\ref{lem:updatealg}). A key insight behind this analysis is that vertices whose eliminators have small rank are not likely to change their eliminators under many updates, allowing us to keep the cost of reindexing $N^-(v)$ small. Another insight is that the maximum degree in the graph induced on vertices whose eliminators have high ranks is small. Therefore, set $N^+(v)$ will be typically small and the fact that it is not indexed by the rank of its members' eliminators is not troublesome.

Theorem~\ref{thm:vertexA} and  Lemma~\ref{lem:updatealg} almost settle our analysis, with the former asserting that the expected number of affected nodes is small, and the latter asserting that the expected time to maintain the data structures, per affected node, is small. One final analytic hurdle emerges at this stage though: These two events are not a priori independent and so the product of the expectations is not an upper bound on the expected running time of an update! To overcome this, we introduce another twist in our analysis. Recall that Theorem~\ref{thm:vertexA} holds even if an entire node is updated (say deleted along with all its edges). When applied to an edge update $(a, b)$, this gives an upper bound of $O(\log n)$ on the expected number of affected vertices even if we condition on any value of $\pi(a)$. (See Lemma~\ref{lem:edgeA}.) The reason, roughly speaking, is that once we condition on $\pi(a)$, the edge update $(a, b)$ can now be regarded as insertion or deletion of vertex $b$. 

Overall, we use the randomization in $\pi(a)$ to bound the expected time per affected vertex by $\polylog n$ and, conditioned on this, still get an $O(\log n)$ upper bound on the expected number of affected vertices due to Lemma~\ref{lem:edgeA}. This allows us to prove an expected $\polylog n$ upper bound on the total running time (see Section~\ref{sec:wrapup}), thus concluding our analysis.

\section{Preliminaries}\label{sec:preliminaries}
\newcommand{\lfmis}[1]{\ensuremath{\mathsf{LFMIS}(#1)}}
\newcommand{\lfmm}[1]{\ensuremath{\mathsf{LFMM}(#1)}}

\newcommand{\eliminator}[3]{\ensuremath{\elim_{#1, #2}(#3)}}
\newcommand{\seliminator}[1]{\ensuremath{\elim(#1)}}

In this section, we formally define {\em lexicographically first MIS} and {\em MM}, mention some of their known properties, and define the notion of {\em eliminators} which are all important for the rest of paper.

\smparagraph{Notation.} For any positive integer $k$, we use $[k]$ to denote set $\{1, \ldots, k\}$. For a graph $G=(V, E)$ and a vertex $v \in V$, we use $N_G(v)$ or, in short, $N(v)$ to denote the set of neighbors of $v$ in $G$ and use $\Gamma(v)$ to denote the set $N(v) \cup \{ v\}$. This notation also extends to any subset  $U$ of $V$ where we use $N(U)$ and $\Gamma(U)$ to respectively denote $\cup_{v \in U} N(v)$ and $\cup_{v \in U} \Gamma(v)$. We also use $\deg_G(v)$ to denote the {\em degree} of vertex $v$, i.e., $\deg_G(v) = |N_G(v)|$.

\smparagraph{Lexicographically First MIS.} The {\em lexicographically first maximal independent set} (LFMIS) of a graph $G=(V, E)$ according to a ranking $\pi: V \to [0, 1]$ over the vertices in $V$ is obtained as follows. Initially, every vertex in $V$ is {\em alive}. We iteratively take the alive vertex $v$ with the minimum rank $\pi(v)$, add $v$ to the MIS, and {\em kill} $v$ and all of its alive neighbors. We use \lfmis{G, \pi} to refer to the subset of vertices that join this MIS. For each vertex $v$, we define the {\em eliminator} of $v$, denoted by $\eliminator{G}{\pi}{v}$, as the (unique) vertex that kills $v$. More precisely, \eliminator{G}{\pi}{v} is the lowest-rank vertex in $(N(v) \cup \{v \}) \cap \lfmis{G, \pi}$. Note that if $v$ is in the MIS, we have $\eliminator{G}{\pi}{v} = v$; otherwise, $\eliminator{G}{\pi}{v} \not= v$ and $\pi(\eliminator{G}{\pi}{v}) < \pi(v)$. When no confusion is possible, we may write $\seliminator{v}$ instead of $\eliminator{G}{\pi}{v}$ for brevity.

\smparagraph{Lexicographically First MM.} All definitions above can be extended to MM as well if we consider LFMIS over the line-graph. The resulting {\em lexicographically first MM} of a graph $G=(V, E)$, which we denote by \lfmm{G, \pi} where $\pi: E \to [0, 1]$ is a ranking over the {\em edges} of $G$ would be as follows. Initially, all the edges are alive. We iteratively pick the alive edge $e$ with the minimum rank $\pi(e)$, add it to the matching, and kill $e$ and all the alive incident edges to $e$. The eliminator $\eliminator{G}{\pi}{e}$ of an edge $e$ in this algorithm can similarly be defined as the lowest-rank edge incident to $e$ (including $e$ itself) that is in the maximal matching \lfmm{G, \pi}.

\smparagraph{LFMIS and LFMM over Random Ranks.}
The two algorithms above are particularly useful when the ranking $\pi$ maps to a random permutation, i.e., each entry of $\pi$ is a real chosen uniformly at random from $[0, 1]$. It is not hard to see that choosing $\Theta(\log n)$ bit reals is enough to guarantee no two entries assume the same rank w.h.p. From now on, when we use the term ``random ranking'' $\pi$, we indeed assume that each entry of $\pi$ has $\Theta(\log n)$ bits.

One useful property of LFMIS over random rankings is that once we, roughly speaking, process $p$ fraction of the vertices with the lowest ranks and remove their MIS nodes and their neighbors, the maximum degree in the remaining graph drops to $O(p^{-1} \cdot \log n)$ w.h.p. The same also holds for LFMM. This property is very well-known \cite{DBLP:conf/spaa/BlellochFS12,DBLP:conf/icml/AhnCGMW15,DBLP:conf/podc/GhaffariGKMR18,DBLP:journals/corr/abs-1901-03744,assadisoda}; when incorporating the definition of eliminators, it would read as follows:

\newcommand{\rlfmis}[2]{\ensuremath{\mathsf{LFMIS}_{#1}(#2)}}
\newcommand{\rlfmm}[2]{\ensuremath{\mathsf{LFMM}_{#1}(#2)}}

\newcommand{\sparseMIS}[0]{Consider a graph $G=(V, E)$, let $\pi: V \to [0, 1]$ be a random ranking, and for any real $p \in [0, 1]$, define $V_p$ as the subset of $V$ including any vertex $v$ with $\pi(\eliminator{G}{\pi}{v}) > p$. W.h.p., for all $O(\log n)$ bit values of $p \in [0, 1]$, the maximum degree in graph $G[V_p]$ is $O(p^{-1} \cdot \log n)$.}
\begin{proposition}\label{prop:sparseMIS}
	\sparseMIS{}
\end{proposition}


\newcommand{\sparseMM}[0]{Consider a graph $G=(V, E)$, let $\pi: E \to [0, 1]$ be a random ranking, and for any real $p \in [0, 1]$, define $E_p$ to be the subset of $E$ including any edge $e$ with $\pi(\eliminator{G}{\pi}{e}) > p$.  W.h.p., for all $O(\log n)$ bit values of $p \in [0, 1]$, every vertex has $O(p^{-1} \cdot \log n)$ incident edges in $E_p$.
}
\begin{proposition}\label{prop:sparseMM}
	\sparseMM{}
\end{proposition}


For the minor differences between the statements of these propositions and those in the literature, we prove them in Appendix~\ref{sec:sparsificationproof}. We emphasize that the changes are straightforward and we claim no novelty on this part.

%
%

%

\section{Fully Dynamic MIS: Data Structures \& The Algorithm}\label{sec:algorithm}
In this section, we present the data structures and the algorithm required for maintaining LFMIS after each update.

We fix a random ranking $\pi$ in the pre-processing step and maintain \lfmis{G, \pi} after each update. Throughout the rest of this section, we focus on the data structures required for maintaining \lfmis{G, \pi} and the algorithm we use to update them. Fix an arbitrary $t$ and suppose that we have to address edge update number $t$. We use ``time $t$'' to refer to the moment {\em after} the first $t$ edge updates. Moreover, we use $G_t = (V, E_t)$ to denote the resulting graph at time $t$.  The following definitions are crucial both for the algorithm's description and its analysis.

\begin{highlighttechnicalwhite}
	\begin{itemize}[leftmargin=15pt, itemsep=0pt]
		\item $\affectedset := \{v \mid \eliminator{G_{t-1}}{\pi}{v} \not= \eliminator{G_t}{\pi}{v} \}$: The set of vertices whose eliminator changes after the update; we call these the {\em affected} vertices.
		\item $\flippedset$: The set of vertices $w$ that belong to exactly one of \lfmis{G_t, \pi} or \lfmis{G_{t-1}, \pi}. We call these the {\em flipped} vertices. Note that $\flippedset \subseteq \affectedset$.
	\end{itemize}
\end{highlighttechnicalwhite}

Our main result in this section is the following algorithm.

\newcommand{\updatealg}[0]{There is an algorithm to update $\lfmis{G, \pi}$ and the data structures required for it after insertion or deletion of any edge $e=(a, b)$ in $O\left(|\affectedset|\min\{\Delta, \frac{\log n}{\min\{ \pi(a), \pi(b) \}} \}\log \Delta\right)$ time w.h.p.
}

\begin{lemma}\label{lem:updatealg}
\updatealg
\end{lemma}

Note that the bound on the update-time in the statement above is parametrized by two random variables $|\affectedset|$ and $\min\{ \pi(a), \pi(b) \}$ of the ranking $\pi$. To provide a concrete bound on the update-time, we need to analyze how these two random variables are related. We prove the necessary tools for this analysis in Section~\ref{sec:analysis} and finally prove that this quantity is in fact $\polylog n$ in Section~\ref{sec:wrapup}. 

In the rest of this section, we only focus on proving Lemma~\ref{lem:updatealg}. We describe the data structures in Section~\ref{sec:datastructures}, describe the algorithm in Section~\ref{sec:algalg}, and prove the correctness and running time of the algorithm in Section~\ref{sec:parametrizedrunningtime}.

\subsection{Data Structures}\label{sec:datastructures}

As described before, our algorithm starts with a pre-processing step where we choose a random ranking $\pi$ over the $n$ fixed vertices in $V$, i.e., as discussed in Section~\ref{sec:preliminaries}, we pick a $\Theta(\log n)$ bit real $\pi(v) \in [0, 1]$ for each vertex $v$. The ranking $\pi$ will then be used to maintain \lfmis{G, \pi} after each update to graph $G$. To update this MIS  efficiently, we maintain the following data structures for each vertex $v \in V$.

\begin{highlighttechnical}
	\begin{itemize}[leftmargin=15pt]
		\item $m(v)$: A binary variable that is 1 if $v \in \lfmis{G, \pi}$ and 0 otherwise.
		\item $k(v)$: The rank of $v$'s eliminator, i.e., $k(v) = \pi(\eliminator{G}{\pi}{v})$. Note that $m(v)=1$ iff $k(v) = \pi(v)$.
		\item $N^-(v)$: The set of neighbors $u$ of $v$ where $k(u) \leq k(v)$. The set $N^-(v)$ is stored as a self-balancing binary search tree (BST) and each vertex $u$ in it is indexed by $k(u)$.
		\item $N^+(v)$: The set of neighbors $u$ of $v$ where $k(u) \geq k(v)$. The set $N^+(v)$ is also stored as a BST, but unlike $N^-(v)$, each member $u$ in $N^+(v)$ is indexed by its ID.
	\end{itemize}
\end{highlighttechnical}

It has to be noted that each vertex $u \in N^-(v)$ is indexed by $k(u)$, a property that may change after an edge update and thus we may need to re-order the vertices in $N^-(v)$. However, the vertices in $N^+(v)$ are simply indexed by their IDs which are static. Also, observe that:

\newcommand{\obsnplusnminus}[0]{For any two neighbors $u$ and $v$, $u \in N^+(v)$ if and only if $v \in N^-(u)$.}

\begin{observation}\label{obs:nplusnminus}
	\obsnplusnminus{}
\end{observation}
\begin{proof}
	If $u \in N^+(v)$, then $k(u) \geq k(v)$; since $N^-(u)$ includes every neighbor $w$ of $u$ with $k(w) \leq k(u)$, and $k(v) \leq k(u)$, we have $v \in N^-(u)$. Similarly, if $v \in N^-(u)$, then $k(v) \leq k(u)$; since $N^+(v)$ includes every neighbor $w$ of $v$ with $k(w) \geq k(v)$ and $k(u) \geq k(v)$, we have $u \in N^+(v)$.
\end{proof}

From now on, we use $m_t(v)$, $k_t(v)$, $N^-_t(v)$ and $N^+_t(v)$ to respectively refer to data structures $m(v)$, $k(v)$, $N^-(v)$ and $N^+(v)$ by time $t$. Before describing the algorithm, we describe the pre-processing step in more details.

\smparagraph{Pre-processing Step.}  Apart from choosing random ranking $\pi$, we initialize an array $\P{v} \gets \emptyset$ for every vertex $v$ in the pre-processing step. This array will later be used in the update algorithm in Section~\ref{sec:algalg}. Moreover, we construct LFMIS over the original graph $G_0=(V, E_0)$ via the trivial approach:  We iterate over the vertices according to $\pi$ to construct \lfmis{G_0, \pi} and set $m(v)$ for each vertex $v$. Then for each vertex $v$, we iterate over all of its neighbors to fill in $k(v)$, $N^+(v)$, and $N^-(v)$. We initially spend $O(n \log n)$ time for sorting the vertices, then for each vertex $v$, we spend $O(\deg(v))$ time to fill in its data structures. This process, overall, takes $O((|V|+|E_0|)\log n)$ time which is clearly optimal (up to a logarithmic factor) as it is required to read the input.

\subsection{The Algorithm}\label{sec:algalg}
We now turn to describe how we maintain the data structures defined in the previous section after each edge update. Consider update number $t$, and suppose that an edge $e=(a, b)$ is either inserted or deleted. Moreover, assume w.l.o.g. that $\pi(a) < \pi(b)$. We show how our data structures can be adjusted accordingly in the time specified by Lemma~\ref{lem:updatealg}.


Since we are maintaining the lexicographically first MIS---and not just any MIS---of a dynamically changing graph, a single edge update can potentially affect vertices that are multiple-hops away. To detect these vertices efficiently, we use an iterative approach with which, intuitively, we do not ``look'' at too many unaffected vertices.   Before formalizing this, we start with an observation. The proof is a simple consequence of the structure of LFMIS and thus we defer it to Section~\ref{sec:proofsMIS}.

\newcommand{\obsinAiff}[0]{For any vertex $v \in \affectedset$, the following properties hold:
\begin{enumerate}[itemsep=0pt,topsep=7pt]
	\item $k_{t-1}(v) \geq \pi(a)$ and $k_{t}(v) \geq \pi(a)$.
	\item if $v \not= b$, then $v$ has a neighbor $u$ such that $\pi(u) < \pi(v)$ and $u \in \flippedset$.
\end{enumerate}
}

\begin{observation}\label{obs:inAiffalowerchanges}
	\obsinAiff
\end{observation}

We start with an intuitive and informal description of the algorithm. The algorithm's formal description and the proofs are given afterwards.

\smparagraph{Algorithm Outline.} Observation~\ref{obs:inAiffalowerchanges} part 2 implies that if a vertex $u$ is in set $\affectedset$, then there should be a path from vertex $b$ to $u$ where all the vertices in the path (except $u$) belong to $\flippedset$ and the ranks in the path are monotonically increasing. This motivates us to use an iterative approach. We start by a set $\S$ which originally only includes vertex $b$. Then we iteratively take the minimum rank vertex $v$ from $\S$, detect whether $v \in \flippedset$ and if so, we add all the ``relevant neighbors'' of $v$ that may continue these monotone paths to set $\S$. Clearly, we cannot add all neighbors of $v$ to $\S$ since there could be as many as $\Omega(\Delta)$ such nodes. Rather, we only consider neighbors $u$ of $v$ where $k_{t-1}(u) \geq \pi(a)$. Observation~\ref{obs:inAiffalowerchanges} part 1 guarantees that every vertex $u \in \affectedset$ has $k_{t-1}(u) \geq \pi(a)$ and thus this set of relevant neighbors is sufficient to ensure any vertex in $\affectedset$ will be added to $\S$ at some point. Note that by definition, for any vertex $u \in V \setminus \affectedset$, both $k(u)$ and $m(u)$ will remain unchanged after the update. Therefore, once we handle all vertices in set \S, for every vertex $u$ in the graph, $k(u)$ and $m(u)$ should be updated. However, note that the adjacency lists of vertices outside \affectedset{} may require to be updated if they have a neighbor in $\affectedset$. We do this at the end of the algorithm. Algorithm~\ref{alg:update} below formalizes the structure of this algorithm and the subroutines used are formalized afterwards.

We use $k_{t-1}(v)$ and $k_t(v)$ to refer to the value $k(v)$ should hold before and after the update respectively. In the process of updating $k(v)$ from $k_{t-1}(v)$ to $k_{t}(v)$, whenever we use $k(v)$ without any subscript in the algorithm, we refer to the value of this data structure at that specific time. In particular, since we update the vertices iteratively, it could happen that in a specific time during the algorithm, for some vertex $u$, $k(u) = k_{t}(u)$ and for another vertex $w$, $k(w) = k_{t-1}(w)$. The same notation extends to $m(v), N^+(v)$, and $N^-(v)$ in the natural way.

\newcommand{\alginitialize}[0]{\textsc{Initialize}\ensuremath{()}}
\newcommand{\updateadjlists}[0]{\textsc{UpdateAdjacencyLists}{\ensuremath{()}}}
\newcommand{\isaffected}[1]{\textsc{IsAffected}\ensuremath{(#1)}}
\newcommand{\updateeliminator}[1]{\textsc{UpdateEliminator}\ensuremath{(#1)}}
\newcommand{\findrelneighbors}[1]{\textsc{FindRelevantNeighbors}\ensuremath{(#1)}}

\begin{tboxalg2e}{Maintaining the data structures after insertion or deletion of an edge $e=(a ,b)$.}
\label{alg:update}
	\begin{algorithm}[H]
	\DontPrintSemicolon
	\SetAlgoSkip{bigskip}
	\SetAlgoInsideSkip{}

	$\S \gets \{b\}$\;
	For each vertex $v$, we have an array $\P{v} = \emptyset$.\tcp*{Initialized in the pre-processing step.}
	\While{\S{} is not empty}{
		Let $v \gets \argmin_{u \in \S} \pi(u)$ be the minimum rank vertex in $\S$.\;
		\If(\tcp*[f]{Checks whether $v \in \affectedset$ in time $O(|\P{v}|)$.}) {\isaffected{v}\label{line:if}}{
			$\mathcal{H}_v \gets \findrelneighbors{v, \pi(a)}$ \;\tcp{$\mathcal{H}_v$ includes neighbors $u$ of $v$ with $k_{t-1}(u) \geq \pi(a)$ and has size $O(\frac{\log n}{\pi(a)})$ w.h.p.} 
			\updateeliminator{v, \mathcal{H}_v}\tcp*{Updates $k(v)$ and $m(v)$ by iterating over $\mathcal{H}_v$.}
			\If{$v \in \flippedset$}{
				\For{any vertex $u \in \mathcal{H}_v$}{
					\lIf{$\pi(u) > \pi(v)$}{
						insert $u$ to \S{} and insert $v$ to $\P{u}$. \label{line:addthings}
					}
				}
			}
		}
		Remove $v$ from $\S{}$ and set $\P{v} \gets \emptyset$.\;
	}
	\updateadjlists{} \tcp*{Updates adjacency lists $N^+$ and $N^-$ where necessary.}\label{line:updateN}
	\end{algorithm}
\end{tboxalg2e}

We use {\em iteration} to refer to iterations of the while loop in Algorithm~\ref{alg:update}. The following invariants hold at the beginning of the algorithm when $\S = \{ b \}$ and, as we will show in Claim~\ref{cl:invshold} via an induction, will continue to hold throughout.

\begin{invariant}\label{inv:S1}
	Consider the start of any iteration and let $v$ be the lowest-rank vertex in $\S$. It holds true that $k(u) = k_t(u)$ and $m(u) = m_t(u)$ for every vertex $u$ with $\pi(u) < \pi(v)$, i.e., $k(u)$ and $m(u)$ already hold the correct values. Moreover, $k(u) = k_{t-1}(u)$ and $m(u) = m_{t-1}(u)$ for every other vertex $u$ with $\pi(u) \geq \pi(v)$.
\end{invariant}

\begin{invariant}\label{inv:S2}
	Consider the start of any iteration and let $v$ be the lowest-rank vertex in $\S$. The set \P{v} includes a vertex $u$ iff: (1) $\pi(u) < \pi(v)$, and (2) $u \in \flippedset$, and (3) $u$ and $v$ are adjacent.
\end{invariant}

\begin{invariant}\label{inv:N}
	For any vertex $u$, before reaching Line~\ref{line:updateN} of Algorithm~\ref{alg:update}, adjacency lists $N^+(u)$ and $N^-(u)$ respectively hold values $N^+_{t-1}(u)$ and $N^-_{t-1}(u)$.
\end{invariant}

We continue by formalizing all the subroutines used in Algorithm~\ref{alg:update}.

\smparagraph{Subroutine {\normalfont \isaffected{v}}.} This function returns true if $v \in \affectedset$ and returns false otherwise. We consider two cases where $v = b$ and $v \not= b$ individually. For the former case, we show that $b \in \affectedset$ if and only if $m(a) = 1$ and $k(b) \geq \pi(a)$. For the latter case, we first scan the set \P{v} to see if there exists a vertex $u \in \P{v}$ with $\pi(u) = k(v)$. If such vertex $u$ exists, then $v \in \affectedset$. Otherwise, let $u$ be the lowest-rank vertex in $\P{v}$ such that $m(u) = 1$. If $\pi(u) < k(v)$, then $v \in \affectedset$ and otherwise $v \not\in \affectedset$. This subroutine clearly takes $O(|\P{v}|)$ time. We also prove its correctness in Claim~\ref{cl:isaffectedcorrectnessandtime}.

\smparagraph{Subroutine {\normalfont \findrelneighbors{v, \pi(a)}}.} The goal in this subroutine is to find the set 
\begin{equation}\label{eq:defH}
	\mathcal{H}_v := \{ u \in N(v) \mid  k_{t-1}(u) \geq \pi(a)\}.
\end{equation}
By definition of $N^+(v)$ and $N^-(v)$, each neighbor $u \in N(v)$ is at least in one of these two sets. Therefore, to construct set $\mathcal{H}_v$, we have to find neighbors $u$ of $v$ with $k_{t-1}(u) \geq \pi(a)$ in both $N^+(v)$ and $N^-(v)$. For the former, we simply iterate over all neighbors $u$ of $v$ in set $N^+(v)$ and if $k_{t-1}(u) \geq \pi(a)$, we add $u$ to $\mathcal{H}_v$. For the latter, recall from Invariant~\ref{inv:N} that $N^-(v) = N^-_{t-1}(v)$; thus, the vertices $u$ in $N^-(v)$ are indexed by $k_{t-1}(u)$. To find only those in $N^-(v)$ with $k_{t-1}(u) \geq \pi(a)$, it suffices to search for index $\pi(a)$ and traverse over all vertices whose index is at least $\pi(a)$. The correctness and an analysis of the running time of this algorithm is provided in Claim~\ref{cl:correctfindrelevantneighbors}.

\smparagraph{Subroutine {\normalfont \updateeliminator{v, \mathcal{H}_v}}.} Given that the minimum-rank vertex $v \in \S$ is in set $\affectedset$, this subroutine updates $k(v)$ assuming that the set $\mathcal{H}_v$ is already computed and given. To do this, let $u$ be the lowest-rank vertex in $\mathcal{H}_v$ with $m(u) = 1$. If no such vertex exists, or if $\pi(u) > \pi(v)$, $v$ has to join the MIS and thus we set $k(v) \gets \pi(v)$ and $m(v) \gets 1$. Otherwise, $u$ has to be the new eliminator of $v$ and we set $k(v) \gets \pi(u)$ and $m(v) \gets 0$. This subroutine clearly takes $O(|\mathcal{H}_v|)$ time. We also prove its correctness in Claim~\ref{claim:updatek}.

\smparagraph{Subroutine {\normalfont \updateadjlists{}}.} If $e$ is deleted, we remove $a$ from $N^+(b)$ and $N^-(b)$, and remove $b$ from $N^+(a)$ and $N^-(a)$ (note that some of these sets may not include the removing vertex). If $e$ is inserted, we insert $a$ and $b$ into each other's ``appropriate'' adjacency list according to the current values of $k(a)$ and $k(b)$; namely:
		\begin{itemize}[topsep=5pt,itemsep=-0.3ex,partopsep=0ex,parsep=1ex]
			\item If $k(a) < k(b)$, insert $a$ into $N^-(b)$, and insert $b$ into $N^+(a)$.
			\item If $k(a) > k(b)$, insert $a$ into $N^+(b)$, and insert $b$ into $N^-(a)$.
			\item If $k(a) = k(b)$, insert $a$ into $N^-(b)$ and $N^+(b)$, and insert $b$ into $N^-(a)$ and $N^+(a)$.
		\end{itemize}
	We also need to update the adjacency lists of any affected vertex $v$, since after changing $k(v)$, some neighbors of $v$ may have to move from $N^+(v)$ to $N^-(v)$ or vice versa. Moreover, if an affected vertex $v$ is in $N^-(u)$ of some vertex $u$, we also need to recompute the position of $v$ in $N^-(u)$, since recall that $v$ should be indexed by $k(v)$ in $N^-(u)$ which has now changed.
	
	To address the changes above, the crucial property is that for any vertex $v \in \affectedset$, any vertex $u$ that has to move between $N^+(v)$ and $N^-(v)$ or has $v$ in its set $N^-(u)$, has to belong to $\mathcal{H}_v$ (see Claim~\ref{cl:adjlistscorrect} for the proof). Therefore in the algorithm, for any vertex $v \in \affectedset$, we only iterate over the vertices $u \in \mathcal{H}_v$ and based on $k(u)$ and $k(v)$, which at this point in the algorithm are correctly updated, determine the membership of vertex $v$ in adjacency lists of vertex $u$ and vice versa. We then update $N^-(v)$, $N^+(v)$, $N^-(u)$ and $N^+(u)$ accordingly.

\subsection{Overview of Correctness \& The (Parametrized) Running Time}\label{sec:parametrizedrunningtime}

The correctness of Algorithm~\ref{alg:update} follows mainly from the greedy structure of LFMIS and does not require a sophisticated analysis. As such, we defer it to Section~\ref{sec:proofsMIS}. Here, we focus on the main ideas required for bounding the running time of the algorithm stated in Lemma~\ref{lem:updatealg}. A complete proof of this lemma is also presented in Section~\ref{sec:proofsMIS}.

One particularly important property is that, w.h.p., the size of set $\mathcal{H}_v$ for every vertex $v \in \affectedset$ is $O\big(\min\{ \Delta, \frac{\log n}{\pi(a)}\}\big)$. This is formally proved in Claim~\ref{cl:correctfindrelevantneighbors} of Section~\ref{sec:proofsMIS}; but the main intuition is as follows. From definition of $\mathcal{H}_v$, every vertex $u \in \mathcal{H}_v$ has $k_{t-1}(u) \geq \pi(a)$. Moreover, since $v \in \affectedset$, by Observation~\ref{obs:inAiffalowerchanges} part 1, we also have $k_{t-1}(v) \geq \pi(a)$. This means that if we construct LFMIS in graph $G_{t-1}$ on the prefix of vertices with rank in $[0, \pi(a))$, then vertex $v$ will survive and will have a remaining degree of at least $|\mathcal{H}_v|$. Since the adversary is oblivious and the ranking $\pi$ and graph $G_{t-1}$ are independently chosen, we can use Proposition~\ref{prop:sparseMIS} to argue that in this remaining graph, maximum degree is, w.h.p., at most $O\big(\min\{ \Delta, \frac{\log n}{\pi(a)}\}\big)$ implying the same upper bound on $|\mathcal{H}_v|$. 

Observe that in the algorithm, only for vertices $v \in \flippedset$ we insert (a subset of) their relevant neighbors $\mathcal{H}_v$ to $\S$. Therefore, the total number of vertices inserted to $\S$ is at most 
$O\big( |\flippedset| \min\{\Delta, \frac{\log n}{\pi(a)} \}\big)$, w.h.p. However, this is not an upper bound on the algorithm's running time since each vertex in $\S$ is not simply processed in constant time. We summarize these procedures below.

\smparagraph{Subroutine {\normalfont \isaffected{v}}.} This subroutine is called for every vertex $v \in \S$. It is clear from description that \isaffected{v} takes $O(|\P{v}|)$ time. Therefore, the aggregated running time of this function for all vertices in $\S$ is $\sum_{v \in S} |\P{v}|$. Observe that each vertex $u \in \P{v}$ is in \flippedset{}. Furthermore, each vertex $u \in \flippedset$ belongs to $\P{v}$ of at most $|\mathcal{H}_u|$ vertices due to Line~\ref{line:addthings}. Therefore, a simple double-counting argument shows that w.h.p., $\sum_{v \in S} |\P{v}| \leq O\big(|\flippedset| \min\{\Delta, \frac{\log n}{\pi(a)} \}\big)$.

\smparagraph{Subroutine {\normalfont \findrelneighbors{v, \pi(a)}}.} This is called for every vertex $v \in \affectedset$. Thanks to the fact that $N^-(v)$ is indexed by $k_{t-1}(.)$ and that $N^+(v)$ has size at most $O(|\mathcal{H}_v|)$ (we show this in the proof of   Claim~\ref{cl:correctfindrelevantneighbors}) this subroutine takes $O(|\mathcal{H}_v| \log\Delta)$ time where the extra $\log \Delta$ factor is for iterating over BST $N^-(v)$. Thus, the aggregated running time is $ O\big(|\affectedset|\min\{\Delta, \frac{\log n}{\pi(a)} \}\log \Delta\big)$.

\smparagraph{Subroutine {\normalfont \updateeliminator{v, \mathcal{H}_v}}.} This subroutine is only called for vertices $v \in \affectedset$ and takes $O(|\mathcal{H}_v|)$ time. Clearly, the aggregated running time is $ O\big(|\affectedset|\min\{\Delta, \frac{\log n}{\pi(a)} \}\big)$, w.h.p. 

\smparagraph{Subroutine {\normalfont \updateadjlists}.} As described in the subroutine, for any vertex $v \in \affectedset$, $v$ has to be re-indexed or moved in adjacency lists of at most $|\mathcal{H}_v|$ of its neighbors. Each such operation requires $O(\log \Delta)$ time. Therefore, the aggregated running time is w.h.p. $ O\big(|\affectedset|\min\{\Delta, \frac{\log n}{\pi(a)} \}\log \Delta\big)$.

The total running time of the algorithm is the sum of the aggregated running time of each of the procedures above which is $O\big(|\affectedset|\min\{\Delta, \frac{\log n}{\pi(a)} \}\log \Delta\big)$ as required by Lemma~\ref{lem:updatealg}.

\section{An Analysis of Affected Vertices: Proof of Theorem~\ref{thm:vertexA}}\label{sec:analysis}

In this section, we prove Theorem~\ref{thm:vertexA} which we briefly highlighted in Section~\ref{sec:highlevel}. In this regard, for any two graphs $G=(V, E)$ and $G'=(V', E')$ with $V' \subseteq V$ and a ranking $\pi$ over $V$, we define $\affectedset_\pi(G, G') := \{v \in V \mid \eliminator{G}{\pi}{v} \not= \eliminator{G'}{\pi}{v} \}$ to be the set of vertices with different eliminators in the two graphs. Note that this is analogous to the definition of ``affected vertices'' in the previous section and hence the choice of notation. A more formal statement of Theorem~\ref{thm:vertexA} reads as follows:

\newcommand{\thmvertexA}[0]{Fix an arbitrary graph $G=(V, E)$ and let $G'=G[V \setminus \{v\}]$ be obtained by removing an arbitrary vertex $v$ from $G$. If $\pi$ is a random ranking over $V$,
	$
	\E_{\pi}[|\affectedset_\pi(G, G')|] \leq O(\log n).
	$
}

\begin{theorem}\label{thm:vertexA}
	\thmvertexA{}
\end{theorem}

The fact that Theorem~\ref{thm:vertexA} bounds the number of affected vertices as a result of a vertex update can be used to bound the affected vertices by $O(\log n)$ as a result of an edge update $e=(a, b)$, {\em even when we condition on any value for $\min\{\pi(a), \pi(b)\}$}.

\newcommand{\lemedgeA}[0]{Fix an arbitrary graph $G=(V, E)$ and let $G'=(V, E')$ be the graph obtained by adding or removing an arbitrary edge $e=(a, b)$ to $G$. If  $\pi$ is a random ranking over $V$,  then for any value of $\lambda \in [0, 1]$, it holds that
	$
	\E_{\pi}[|\affectedset_\pi(G, G')| \mid \min\{\pi(a), \pi(b)\}=\lambda] \leq O(\log n).
	$
}

\begin{lemma}\label{lem:edgeA}
	\lemedgeA{}
\end{lemma}

Lemma~\ref{lem:edgeA} is crucial for our analysis as it implies that the two random variables $|\affectedset_\pi(G, G')|$ and $\min\{\pi(a), \pi(b)\}$, which recall are used in the statement of Lemma~\ref{lem:updatealg}, can be regarded as ``almost'' independent. We elaborate more on this in Section~\ref{sec:wrapup}. 

We first prove Lemma~\ref{lem:edgeA} given the correctness of Theorem~\ref{thm:vertexA}. The bulk of analysis is then concentrated around proving Theorem~\ref{thm:vertexA}.

\begin{proof}[Proof of Lemma~\ref{lem:edgeA}]
	Suppose w.l.o.g. that $\pi(a) < \pi(b)$, i.e., $\pi(a) = \lambda$ by the conditional event. Let $U$ be the subset of $V$ containing vertices $w$ with $\pi(w) \leq \pi(a)$. We prove the lemma even when the set $U$, and the rank of the vertices in it are chosen adversarially. 
	
	We have $G'[U] = G[U]$ since $u \not\in U$ and the only difference of the two graphs $G'$ and $G$ which is in edge $(a, b)$ does not belong to either of the two induced graphs. Therefore, we have $\lfmis{G'[U], \pi} = \lfmis{G[U], \pi}$; let $I_U$ be this MIS. Furthermore, let $H'(V'_H, E'_H)$ and $H(V_H, E_H)$ be the residual graphs after we remove vertices in $I_U$ and their neighbors from $G'$ and $G$ respectively. It is not hard to see that either $H = H'$ (if $a \not\in I_U$ or if $b$ has another neighbor $w$ with $\pi(w) < \pi(a)$ in $I_U$) or $H'$ has exactly one extra vertex than $H$ which has to be $b$, i.e., $H = H'[V'_H \setminus \{ b \}]$. In the former case, since the two graphs are equal, no matter how $\pi$ is chosen, the eliminators of all vertices will be the same. In the latter case, the two graphs $H$ and $H'$ differ in only one vertex and no information about the relative order of the vertices in $V_H$ or $V'_H$ in $\pi$ is revealed. Therefore, by Theorem~\ref{thm:vertexA}, the expected number of vertices whose eliminators are different in $H$ and $H'$ is at most $O(\log n)$.
\end{proof}

\newcommand{\parent}[2]{\ensuremath{p_{#1}(#2)}}
\renewcommand{\path}[2]{\ensuremath{P_{#1}(#2)}}

We now, turn to prove Theorem~\ref{thm:vertexA} and start with some notation. Throughout the rest of this section, vertex $v$ should be regarded as fixed. We use $I$ and $I'$ to respectively denote independent sets \lfmis{G, \pi} and \lfmis{G', \pi}. Also, for brevity, we use $\affectedset_\pi$ instead of $\affectedset_\pi(G, G')$. Furthermore, we define $\flippedset_{\pi}$ as the subset of vertices in $\affectedset_{\pi}$ whose MIS-status is \underline{f}lipped, i.e., $u \in \flippedset_{\pi}$ if and only if $u$ belongs to exactly one of $I$ or $I'$.

Instead of rankings, it will be more convenient to consider permutations for the arguments of this section. That is, we assume that a permutation $\pi: V \to [n]$ from the set $\Pi$ of all $n!$ possible permutations is drawn uniformly at random and the LFMIS is constructed according to this permutation. It is clear that LFMIS according to a random rank follows exactly the same distribution as that according to a random permutation.\footnote{To see this, observe that to draw a random permutation $\pi: V \to [n]$, one can first draw a random rank $\rho: V \to [0, 1]$ and then sort the vertices based on $\rho$.}

The following observation is very similar to Observation~\ref{obs:inAiffalowerchanges} of the previous section and will be very useful here too.

\newcommand{\obseveryvertexhasneighborinF}[0]{If $\affectedset_\pi$ is non-empty, then $v \in I$ and $v \in \flippedset_\pi$. Furthermore, for every vertex $u \in \affectedset_\pi \setminus \{ v \}$, there is another vertex $w \in \flippedset_\pi$ that is adjacent to $u$ and $\pi(w) < \pi(u)$.}

\begin{observation}\label{obs:everyvertexhasneighborinF}
	\obseveryvertexhasneighborinF{}
\end{observation}
\begin{proof}
	We first prove that if $\affectedset_\pi \not= \emptyset$ then $v \in I$. Assume for contradiction that $v \not\in I$ and $\affectedset_\pi \not= \emptyset$. Since $v$ does not belong to $G'$, we also have $v \not\in I'$, i.e., $v$ is in neither of the two maximal independent sets $I$ and $I'$. Now take the minimum rank vertex $u$ in $\affectedset_\pi$ (which exists since $\affectedset_\pi \not= \emptyset$). Since $u \in \affectedset_\pi$, by definition, its eliminators should be different in $I$ and $I'$. Therefore, there should exist a vertex $w$ with $\pi(w) < \pi(u)$, that is in exactly one of the two maximal independent sets. Since $v$ is in neither of $I$ and $I'$, $w \not= v$. However, in this case, $w$ would also belong to $\affectedset_\pi$, contradicting that $u$ is the minimum rank vertex in $\affectedset_\pi$, and completing the proof of this part.
	
	For the second part, fix a vertex $u \in \affectedset_\pi$ and let $x$ and $x'$ be its eliminators in $I$ and $I'$ respectively. Note that $x$ and $x'$ cannot be the same vertex or otherwise $u \not\in \affectedset_\pi$. Suppose that $\pi(x) < \pi(x')$. The fact that $x$ is an eliminator of $u$ in $I$ means that $x \in I$. On the other hand, the fact that $x'$, instead of $x$, is the eliminator of $u$ in $I'$ means that $x \not\in I'$. This means that $x$ has to belong to $\flippedset_\pi$. A similar argument holds for the case where $\pi(x') < \pi(x)$.
\end{proof}

For a vertex $u \in \affectedset_\pi \setminus \{ v \}$, we define the {\em parent} of $u$, denoted by $\parent{\pi}{u}$, as its neighbor in $\flippedset_\pi$ (which exists by observation above) with the lowest rank, i.e., $\parent{\pi}{u} = \argmin_{w \in N(u) \cap \flippedset_\pi} \pi(w)$. Furthermore, we define the {\em propagation path} $\path{\pi}{u}$ of each vertex $u \in \affectedset_\pi$ as:
$$
\path{\pi}{u} = \begin{cases}
	(v) & \text{if $u = v$,}\\
	(\path{\pi}{\parent{\pi}{u}}, u) & \text{otherwise.}
\end{cases}
$$
With a slight abuse of notation, $\path{\pi}{u}$ can be denoted by a sequence $(w_1, \ldots, w_k)$ where $w_1 = v$, $w_k = u$, and for every $i \in [k-1]$, $w_i = \parent{\pi}{w_{i+1}}$. Note that this sequence is a valid path of the graph because by definition each vertex is a neighbor of its parent and $\pi(\parent{\pi}{u})$ is strictly smaller than $\pi(u)$ by Observation~\ref{obs:everyvertexhasneighborinF}, thus, no vertex can be visited twice in the sequence. Furthermore, $w_1 = v$ because every vertex $w \in \affectedset_\pi$ has a parent $\parent{\pi}{w}$ except $v$.

\begin{claim}\label{cl:alternateinMIS}
	Fix an arbitrary permutation $\pi$, an arbitrary vertex $u \in \affectedset_\pi$, and let $P_\pi(u) = (w_1, \ldots, w_k)$. For odd $i \in [k-1]$, $w_i \in \lfmis{G, \pi}$ and for even $i \in [k-1]$, $w_i \not\in \lfmis{G, \pi}$.
\end{claim}
\begin{proof}
	Since $u \in \affectedset_\pi$ and thus $\affectedset_\pi \not= \emptyset$, we already know from Observation~\ref{obs:everyvertexhasneighborinF} that vertex $v = w_1$ has to belong to \lfmis{G, \pi}, proving the claim for $i=1$. To complete the proof, we show that for any $i \in [k-2]$, exactly one of $w_i$ and $w_{i+1}$ is in $\lfmis{G, \pi}$.
	
	First, observe that since $\lfmis{G, \pi}$ is an independent set, no two adjacent vertices can belong to it. Therefore, we only have to show that for any $i \in [k-2]$, it cannot be the case that neither of $w_i$ and $w_{i-1}$ are in \lfmis{G, \pi}. Suppose for contradiction that this holds. By definition of propagation paths, and since $i \in [k-2]$, we get that $w_i$ is the parent of $w_{i+1}$ and $w_{i+1}$ is the parent of $w_{i+2}$. Every vertex that is a parent of another vertex has to be in $\flippedset_\pi$ by definition. Therefore, both $w_i$ and $w_{i+1}$ belong to $\flippedset_\pi$. Combined with the assumption that neither of $w_i$ and $w_{i+1}$ are in \lfmis{G, \pi}, both have to belong to \lfmis{G', \pi} (by definition of $\flippedset_\pi$) which cannot be possible since \lfmis{G', \pi} is also an independent set.
\end{proof}

Let $\Pi$ denote the set of all permutations over $V$. We say a permutation $\pi \in \Pi$ is {\em unlikely}, if for some vertex $u \in V$, $|P_\pi(u)| > \beta \log n$ where $\beta$ is a constant that we fix later, and {\em likely} otherwise. Denoting the set of likely and unlikely permutations by $\Pi_L$ and $\Pi_U$ respectively, we have
\begin{equation}\label{eq:832921234}
\E_\pi[|\affectedset_\pi|] = \Pr[\pi \in \Pi_L] \cdot \E_{\pi \sim \Pi_L}[|\affectedset_\pi|] + \Pr[\pi \in \Pi_U] \cdot \E_{\pi \sim \Pi_U}[|\affectedset_\pi|].
\end{equation}
We prove $\E_\pi[|\affectedset_\pi|] = O(\log n)$ by bounding the two terms in (\ref{eq:832921234}) individually.

\begin{lemma}[\cite{DBLP:conf/spaa/BlellochFS12,DBLP:conf/soda/FischerN18}]\label{lem:unlikeliesareunlikely}
	If $\beta$ is a large enough constant, $\Pr[\pi \in \Pi_U] \leq n^{-2}$.
\end{lemma}

\begin{lemma}[likely permutations]\label{lem:likely}
	$\E_{\pi \sim \Pi_L}[|\affectedset_\pi|] = O(\log n)$.
\end{lemma}

Lemma~\ref{lem:unlikeliesareunlikely} almost directly follows from the earlier works of \cite{DBLP:conf/spaa/BlellochFS12,DBLP:conf/soda/FischerN18} on bounding parallel round complexity of LFMIS over a random permutation; we provide the details in  Section~\ref{sec:unlikely}. Lemma~\ref{lem:likely}, which is proven in Section~\ref{sec:likely}, constitutes the novel part of the proof and is indeed where bulk of the whole analysis is concentrated on. Below, we first show why Lemmas~\ref{lem:likely} and \ref{lem:unlikeliesareunlikely} are sufficient to prove Theorem~\ref{thm:vertexA}.

\begin{proof}[Proof of Theorem~\ref{thm:vertexA}]
	By Lemma~\ref{lem:likely}, we have $\E_{\pi \sim \Pi_L}[|\affectedset_\pi|] = O(\log n)$. Since $\Pr[\pi \in \Pi_L] \leq 1$ for being a probability, we get $\Pr[\pi \in \Pi_L] \cdot \E_{\pi \sim \Pi_L}[|\affectedset_\pi|] \leq O(\log n)$, i.e., the first term in (\ref{eq:832921234}) is bounded by $O(\log n)$. On the other hand, by Lemma~\ref{lem:unlikeliesareunlikely}, we have $\Pr[\pi \in \Pi_U] \leq n^{-2}$. Using this, we can bound the second term in (\ref{eq:832921234}) to be as small as $n^{-1}$ even if $\affectedset_\pi$ includes all $n$ vertices for any $\pi \in \Pi_U$. Therefore overall, we get $\E_\pi[|\affectedset_\pi|] \leq O(\log n) + n^{-1} = O(\log n)$, which is the desired bound.
\end{proof}

\subsection{Handling Likely Permutations: Proof of Lemma~\ref{lem:likely}}\label{sec:likely}

In the rest of this section, we focus on proving Lemma~\ref{lem:likely}. The overall plan is as follows. For each permutation $\pi \in \Pi_L$, we {\em blame} a set of permutations $B(\pi) \subseteq \Pi$ such that:
\begin{enumerate}[label={(P\arabic*)}, ref=P\arabic*, itemsep=0pt]
	\item $|B(\pi)| \geq |\affectedset_\pi|$.\label{prop:blamesatleastD}
	\item For each permutation $\pi' \in \Pi$, there are at most $\beta\log n$ permutations $\pi \in \Pi_L$ where $\pi' \in B(\pi)$.\label{prop:notblamedbymuch}
\end{enumerate}

We first prove that having such blaming sets satisfying properties \ref{prop:blamesatleastD} and \ref{prop:notblamedbymuch} is sufficient for proving Lemma~\ref{lem:likely} and then describe how the blaming sets are constructed.

\begin{proof}[Proof of Lemma~\ref{lem:likely}]
	Defining $X$ as the sum $\sum_{\pi \in \Pi_L} |\affectedset_\pi|$, we have:
	\begin{equation}\label{eq:1000}
	\E_{\pi \sim \Pi_L}[|\affectedset_\pi|] = \sum_{\pi \in \Pi_L} \Pr[\text{drawing  $\pi$} \mid \pi \in \Pi_L] \cdot |\affectedset_\pi| = \frac{1}{|\Pi_L|} \sum_{\pi \in \Pi_L} |\affectedset_\pi|= \frac{X}{|\Pi_L|}.
	\end{equation}
	By property \ref{prop:blamesatleastD}, $|B(\pi)| \geq |\affectedset_\pi|$ for every $\pi \in \Pi_L$. Thus, $$Y := \sum_{\pi \in \Pi_L} |B(\pi)| \geq \sum_{\pi \in \Pi_L} |\affectedset_\pi| = X.$$
	On the other hand, since by property \ref{prop:notblamedbymuch}, each permutation $\pi' \in \Pi$ belongs to $B(\pi)$ of at most $\beta\log n$ other permutations $\pi$, a simple double counting argument gives $Y \leq |\Pi|\beta \log n$; implying also that $X \leq |\Pi|\beta \log n$. Moreover, since $\Pi_L = \Pi \setminus \Pi_U$ and by Lemma~\ref{lem:unlikeliesareunlikely}, $\frac{|\Pi_U|}{|\Pi|} < n^{-2}$, it holds that $\frac{|\Pi_L|}{|\Pi|} > 1- n^{-2}$, thus, $|\Pi| = O(|\Pi_L|)$. For this, $X \leq |\Pi|\beta \log n$ implies $X = O(|\Pi_L| \log n)$. Plugging this into (\ref{eq:1000}), we get $\E_{\pi \sim \Pi_L}[|\affectedset_\pi|] \leq \frac{O(|\Pi_L|\log n)}{|\Pi_L|}  = O(\log n)$ as desired.
\end{proof}

For every permutation $\pi \in \Pi_L$, and each vertex $u \in \affectedset_\pi$, we construct a permutation $\varphi_{\pi, u} \in \Pi$. The blaming set of $\pi$ will then be the set $B(\pi) = \bigcup_{u \in \affectedset_\pi} \{ \varphi_{\pi, u} \}$. For a vertex $u \in \affectedset_\pi$, with $P_\pi(u) = (v=w_1, w_{2}, \ldots, w_k = u)$, we construct permutation $\varphi_{\pi, u}$ as follows:
\begin{highlighttechnicalwhite}
	\begin{itemize}[leftmargin=15pt, itemsep=0pt]
	\item For each vertex $w \not\in P_\pi(u)$, $\varphi_{\pi, u}(w) \gets \pi(w)$.
	\item $\varphi_{\pi, u}(w_1) \gets \pi(w_k)$.
	\item For any $2 \leq i \leq k$, $\varphi_{\pi, u}(w_i) \gets \pi(w_{i-1})$.
\end{itemize}
\end{highlighttechnicalwhite}

In other words, permutation $\varphi_{\pi, u}$ on all vertices outside $P_\pi(u)$ is exactly the same as $\pi$, however for the vertices in $P_\pi(u)$, $\varphi_{\pi, u}$ is obtained by rotating the $\pi$ ranks by one index towards $u$. An example is shown in the figure below.

\begin{figure}[h]
  \centering
  \includegraphics{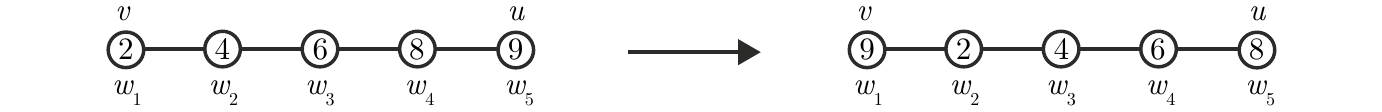}
\end{figure}

As captured by the following observation, it is not hard to show that with this construction, property \ref{prop:blamesatleastD} is indeed satisfied:

\begin{observation}\label{obs:p1satisfied}
	By construction above, property \ref{prop:blamesatleastD} is satisfied.
\end{observation}
\begin{proof}
	The reason is that we indeed construct $|\affectedset_\pi|$ permutations to include in $B(\pi)$: $\varphi_{\pi, u}$ for each $u \in \affectedset_\pi$. Note, however, that we still need to argue that for any two vertices $u$ and $w$ in $\affectedset_\pi$, permutations $\varphi_{\pi, u}$ and $\varphi_{\pi, w}$ are not the same so that the set containing them has size $|\affectedset_\pi|$. This follows because $\varphi_{\pi, w}(v) = \pi(w)$ and $\varphi_{\pi, u}(v) = \pi(u)$ but $\pi(u) \not= \pi(w)$, implying that $\varphi_{\pi, w}(v) \not= \varphi_{\pi, u}(v)$ and thus the two permutations $\varphi_{\pi, w}$ and $\varphi_{\pi, u}$ are not equal.
\end{proof}

The harder part is to show that our construction also satisfies property \ref{prop:notblamedbymuch}:

\begin{claim}\label{cl:p2satisfied}
	By construction above, property~\ref{prop:notblamedbymuch} is also satisfied. 
\end{claim}

Suppose that a permutation $\rho$ is blamed by permutations $\pi$ and $\pi'$, i.e., $\rho \in B(\pi) \cap B(\pi')$. This means that there should exist vertices $u \in \affectedset_\pi$ and $u' \in \affectedset_{\pi'}$ where $\varphi_{\pi, u} = \varphi_{\pi', u'} = \rho$. To prove Claim~\ref{cl:p2satisfied}, we analyze the circumstances under which this may occur. Consider the propagation paths $P_\pi(u) = (w_1, w_{2}, \ldots, w_k)$ and $P_{\pi'}(u') = (w'_{1}, w'_{2}, \ldots, w'_{k'})$ and recall that $w_k = u$, $w'_{k'} = u'$ and $w_1 = w'_1 = v$. Let $j$ be the largest integer where for any $i \in \{1, \ldots, j\}$, we have $w_i = w'_i$. Note that clearly $j \geq 1$ since $w'_1 = w_1 = v$. We call $w_j$ (or equivalently $w'_j$) {\em the branching vertex} and analyze the following scenarios which cover all possibilities individually (see Figure~\ref{fig:branchingvertex}):
\begin{itemize}
	\item \textbf{Scenario 1:} $j$ is odd, $w_j \not= u$, and $w_j \not= u'$.
	\item \textbf{Scenario 2:} $j$ is even, $w_j \not= u$, and $w_j \not= u'$.
	\item \textbf{Scenario 3:} at least one of $u$ or $u'$ is the same as $w_j$.
\end{itemize}

\begin{figure}[h]
  \includegraphics[width=\textwidth]{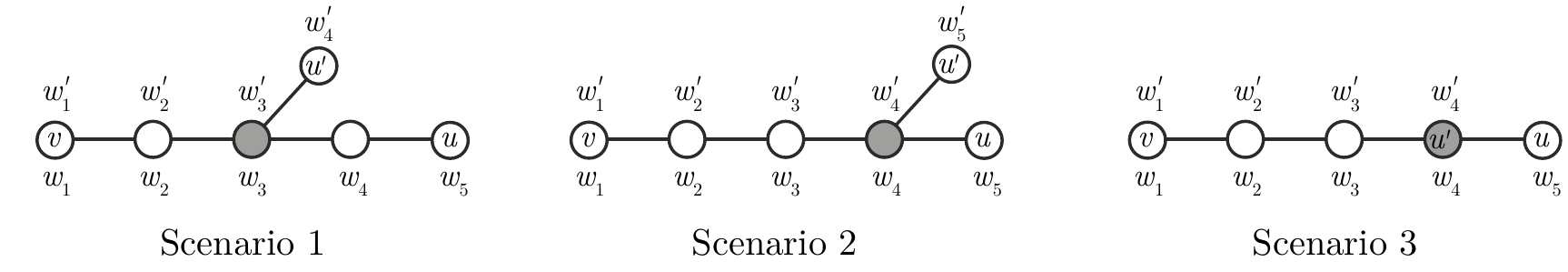}
  \caption{The grey vertex in each scenario, denotes the corresponding branching vertex $w_j$.}
  \label{fig:branchingvertex}
\end{figure}


The claim below unveils several important structural properties of propagation paths and will be our main tool to prove Claim~\ref{cl:p2satisfied}. See Figure~\ref{fig:props} for an illustration of some of these properties.

\begin{claim}\label{cl:propofequalperms}
	Consider two different permutations $\pi$ and $\pi'$ in $\Pi_L$ and two (possibly the same) vertices $u$ and $u'$.	Let $w_j$ be the branching vertex for propagation paths $P_\pi(u)=(w_1, \ldots, w_k)$ and $P_{\pi'}(u') = (w'_{1}, \ldots, w'_{k'})$. If $\varphi_{\pi, u} = \varphi_{\pi',u'}$ and $\pi'(w_j) \geq \pi(w_j)$, then:
	\begin{enumerate}[itemsep=0pt,topsep=5pt]
		\item for every vertex $w$ that does not belong to either of $P_\pi(u)$ and $P_{\pi'}(u')$, $\pi(w)=\pi'(w)$.\label{prop:alloutsideequalrank}
		\item $\pi(w) = \pi'(w)$ for every vertex $w$ with $\pi(w) < \pi(w_j)$.\label{prop:allbelowpiwjareequal}
		\item $\pi(w_k) = \pi'(w'_{k'})$.\label{prop:piuandpipupareeq}
		\item $k \geq j+1$ (i.e., vertex $w_{j+1}$ should exist) and $\pi'(w_{j+1}) = \pi(w_j)$.\label{prop:wjp1existsandetc}
		\item $w_{j+1} \in \lfmis{G, \pi'}$.\label{prop:wjp1inmis}
	\end{enumerate}
\end{claim}

\begin{figure}[h]
  \centering
  \includegraphics[scale=1]{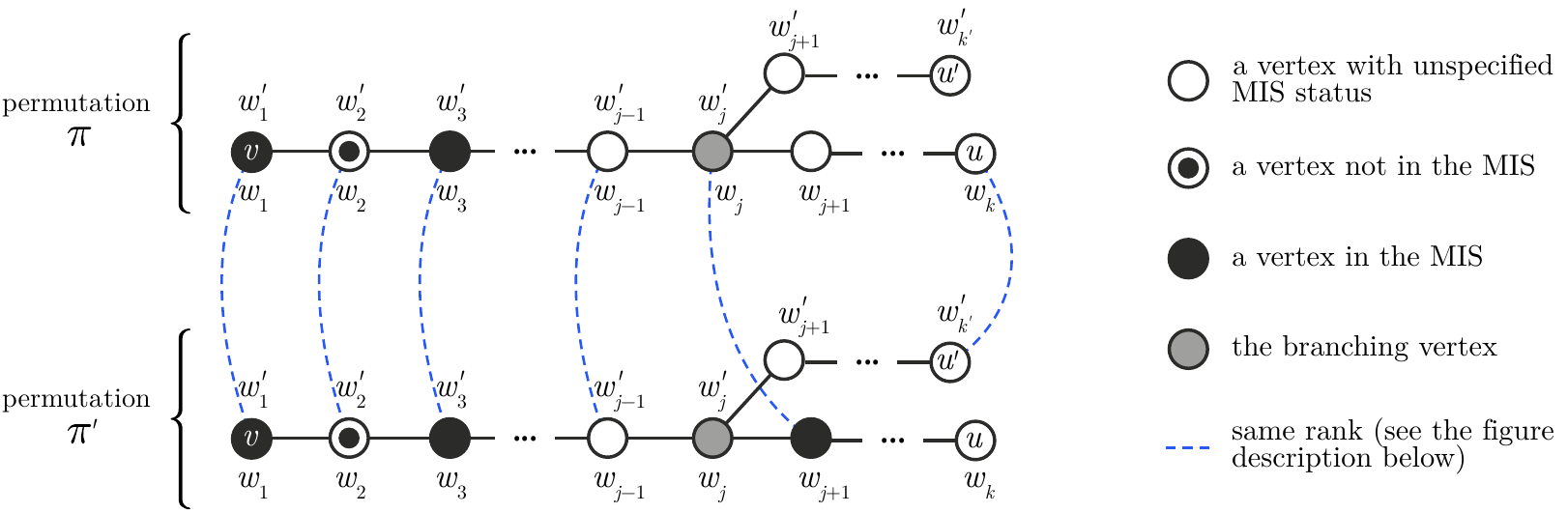}
  \caption{Illustration of some of the properties obtained from Claims~\ref{cl:alternateinMIS} and \ref{cl:propofequalperms} for vertices in $P_\pi(u)=(w_1, \ldots, w_k)$ and $P_{\pi'}(u')=(w'_1, \ldots, w'_{k'})$ given that $\varphi_{\pi, u} = \varphi_{\pi',u'}$ and $\pi'(w_j) \geq \pi(w_j)$ where $w_j$ is the branching vertex. A dashed line between vertices $x$ on the top and $y$ on the bottom implies $\pi(x)=\pi'(y)$. Note that for the illustration purpose, this figure models only scenarios 1 and 2; however, Claims~\ref{cl:alternateinMIS} and \ref{cl:propofequalperms} are general and hold for all three scenarios.}
  \label{fig:props}
\end{figure}

We first show how these properties can be used to prove Claim~\ref{cl:p2satisfied}, then prove Claim~\ref{cl:propofequalperms}.

\begin{proof}[Proof of Claim~\ref{cl:p2satisfied}] 
Suppose that a permutation $\rho$ is blamed by two permutations $\pi$ and $\pi'$, and let $u$ and $u'$ be the vertices where $\varphi_{\pi, u} = \varphi_{\pi',u'}$ (we note that $u$ and $u'$ may be the same vertex). We show that if these assumptions hold, then scenarios 1 and 2 defined above would lead to contradictions, implying that scenario 3 is the only case for which this may occur. To show this, we assume w.l.o.g. that $\pi'(w_j) \geq \pi(w_j)$ so that all conditions of Claim~\ref{cl:propofequalperms} are satisfied.
\begin{description}
	\item[Scenario 1.] Since in this scenario $w_j \not= u'$, we get $j < k'$; more precisely, $j \in [k'-1]$. Furthermore, recall that $j$ is assumed to be odd in scenario 1. Combining the two conditions,  Claim~\ref{cl:alternateinMIS} implies $w_j \in \lfmis{G, \pi'}$. On the other hand, by Claim~\ref{cl:propofequalperms} part \ref{prop:wjp1inmis}, $w_{j+1} \in \lfmis{G, \pi'}$. However, this is a contradiction since by definition of $P_\pi(u)$, $w_j = \parent{\pi}{w_{j+1}}$ and thus $w_j$ and $w_{j+1}$ are neighbors; meaning that they both cannot belong to independent set $\lfmis{G, \pi'}$.
	\item[Scenario 2.] The assumption $w_j \not= u'$ implies that there is a vertex $w'_{j+1}$ and that $w_j = \parent{\pi'}{w'_{j+1}}$; thus by definition of $\parent{\pi'}{w'_{j+1}}$, $w_j \in \flippedset_{\pi'}$. It also implies that $j \in [k'-1]$ (as argued in scenario 1). But since $j$ is even in this scenario, Claim~\ref{cl:alternateinMIS} implies $w_j \not\in \lfmis{G, \pi'}$. Let us use $H$ to denote the graph $G[V \setminus \{v\}]$ obtained by removing vertex $v$ from $G$. Recall that by definition, $w_j$ is in $F_{\pi'}$ iff its MIS-status is different in $\lfmis{G, \pi'}$ and $\lfmis{H, \pi'}$. Therefore, since $w_j \not\in \lfmis{G, \pi'}$ we have to have $w_j \in \lfmis{H, \pi'}$. This also implies that $w_{j+1} \not\in \lfmis{H, \pi'}$ since as argued in scenario 1, $w_j$ and $w_{j+1}$ are neighbors in $G$ and thus $H$. On the other hand, similar to scenario 1, we should have $w_{j+1} \in \lfmis{G, \pi'}$ by Claim~\ref{cl:propofequalperms} part~\ref{prop:wjp1inmis}. Therefore, since $w_{j+1}$ has a different MIS-status in $\lfmis{G, \pi'}$ and $\lfmis{H, \pi'}$, we have $w_{j+1} \in \flippedset_{\pi'}$. By definition, $\flippedset_{\pi'} \subseteq \affectedset_{\pi'}$ thus by Observation~\ref{obs:everyvertexhasneighborinF},
	\begin{equation*}
		\text{ there exists a vertex } x \in N_H(w_{j+1}) \text{ such that } x \in \flippedset_{\pi'} \text{ and } \pi'(x) < \pi'(w_{j+1}).
	\end{equation*}
	Furthermore, by Claim~\ref{cl:propofequalperms} part \ref{prop:wjp1existsandetc}, $\pi'(w_{j+1}) = \pi(w_j)$; combined with $\pi'(x) < \pi'(w_{j+1})$, this implies $\pi'(x) < \pi(w_j)$. Observe that by Claim~\ref{cl:propofequalperms} part \ref{prop:allbelowpiwjareequal} the two permutations $\pi$ and $\pi'$ are exactly the same on the set of vertices with rank less than $\pi(w_j)$. Therefore, $x \in \lfmis{G, \pi'}$ iff $x \in \lfmis{G, \pi}$, and $x \in \lfmis{H, \pi'}$ iff $x \in \lfmis{H, \pi}$. As a result, $x \in \flippedset_{\pi'}$ implies that $x \in \flippedset_\pi$.
	
	Finally, recall that $\parent{\pi}{w_{j+1}}$ is by definition the lowest-rank neighbor of $w_{j+1}$ in $\flippedset_\pi$. Therefore, since $x \in \flippedset_\pi$ and $\pi(x) < \pi(w_j)$, we have $\parent{\pi}{w_{j+1}} \not= w_j$. This contradicts the definition of $P_\pi(u)$ which guarantees $w_j = \parent{\pi}{w_{j+1}}$.
\end{description}
As shown above, the only case for which we might get $\varphi_{\pi,u} = \varphi_{\pi',u'}$ is scenario 3 as the other two scenarios lead to contradictions. We now show that because of the very specific structure of scenario 3, each permutation is blamed by at most $\beta \log n$ permutations.

Fix a permutation $\rho$ and let $C_\rho$ be a set that includes every pair $(\pi, u)$ with $\pi \in \Pi_L$ and $u \in V$ for which $\varphi_{\pi, u} = \rho$. Clearly, $|C_\rho|$ is an upper bound on the number of permutations that blame $\rho$, thus it suffices to bound $|C_\rho|$ by $\beta \log n$.

First, we show that for any two different pairs $(\pi, u)$ and $(\pi', u')$ in $C_\rho$, we have $|P_\pi(u)| \not= |P_{\pi'}(u')|$. Suppose for the sake of contradiction that $|P_\pi(u)| = |P_{\pi'}(u')|$. Let $P_\pi(u) = (w_1, \ldots, w_k)$ and $P_{\pi'}(u') = (w'_1, \ldots, w'_{k})$ be the vertices in the two paths and let $w_j$ be the branching vertex. We know that $\varphi_{\pi, u} = \varphi_{\pi',u'} = \rho$ since the pairs belong to $C_\rho$. Therefore, scenario 3 has to occur and thus either $w_j = u'$ or $w_j = u$. In either case, we get $j = k$ since $u = w_k$ and $u' = w'_k$. Furthermore, by definition of the branching vertex we have $w_i = w'_i$ for any $i \in [j]$. Moreover, by Claim~\ref{cl:propofequalperms} parts~\ref{prop:allbelowpiwjareequal} and \ref{prop:piuandpipupareeq}, for any $i \in [j]$, we have $\pi(w_i)=\pi'(w'_i)$. Meaning that the set of vertices and their ranks in the two permutations are exactly the same on the propagation paths. On the other hand, for any vertex $x$ that does not belong to the propagation paths, we also have $\pi(x) = \pi'(x)$ due to Claim~\ref{cl:propofequalperms} part~\ref{prop:alloutsideequalrank}. Combining these, we get that $\pi = \pi'$. We also showed that $w_k = w'_k$ and thus $u = u'$. Therefore, the two pairs $(\pi, u)$ and $(\pi', u')$ are identical, which is in contradiction with our initial assumption that they are different.

Now we show that $|C_\rho| \leq \beta\log n$. Suppose for contradiction that there are at least $\beta\log n + 1$ pairs in $C_\rho$. As shown in the previous paragraph, for each pair $(\pi, u) \in C_\rho$, $|P_\pi(u)|$ is unique. Therefore, if $|C_\rho| > \beta\log n$, there should be at least a pair $(\pi, u)$ with $|P_\pi(u)| \geq \beta\log n + 1$. However, by definition, the propagation-path of every vertex in every permutation $\pi \in \Pi_L$, has size at most $\beta\log n$ which is a contradiction. Therefore, $|C_\rho| \leq \beta\log n$ for any permutation $\rho$, thus every permutation $\rho$ is blamed by at most $\beta\log n$ other permutations. This means that property \ref{prop:notblamedbymuch} is also satisfied by our mapping, as desired.
\end{proof}

\subsection{The Mapping's Structural Properties: Proof of Claim~\ref{cl:propofequalperms}}
In what follows, we prove the parts of Claim~\ref{cl:propofequalperms} one by one. Note that the proof of each part may depend on the correctness of the previous parts.

\begin{proof}[Proof of Claim~\ref{cl:propofequalperms} part~\ref{prop:alloutsideequalrank}]
	For any vertex $w$ that does not belong to the propagation paths $P_\pi(u)$ and $P_{\pi'}(u')$, we have $\varphi_{\pi, u}(w) = \pi(w)$ and $\varphi_{\pi',u'}(w) = \pi'(w)$ by construction of  permutations $\varphi_{\pi,u}$ and $\varphi_{\pi',u'}$; hence, to have $\varphi_{\pi, u} = \varphi_{\pi', u'}$ it should hold that $\pi(w) = \pi'(w)$.
\end{proof}
	
\begin{proof}[Proof of Claim~\ref{cl:propofequalperms} part~\ref{prop:allbelowpiwjareequal}.] Consider a vertex $w$ with $\pi(w) < \pi(w_j)$, we prove $\pi(w)=\pi'(w)$.
\begin{description}[itemsep=0pt]
	\item[Case 1: $w \not\in P_\pi(u)$ and $w \not\in P_{\pi'}(u')$.]
		In this case, by Claim~\ref{cl:propofequalperms} part~\ref{prop:alloutsideequalrank} we have $\pi(w) = \pi'(w)$.
	\item[Case 2: $w \in P_\pi(u)$.]
	By definition of propagation-path $P_{\pi}(u)$, we have $\pi(w_1) < \ldots < \pi(w_k)$. Therefore, since $w \in P_\pi(u)$ and $\pi(w) < \pi(w_j)$, we should have $w = w_i$ for some $i < j$. Since $w_j$ is the branching vertex, this means $w_{i+1} = w'_{i+1}$ and $w_i = w'_i$. By construction of $\varphi_{\pi,u}$ and $\varphi_{\pi',u'}$, we have $\varphi_{\pi,u}(w_{i+1}) = \pi(w_i)$ and $\varphi_{\pi',u'}(w'_{i+1}) = \pi'(w'_i)$. Combined with $w_{i+1} = w'_{i+1}$, $w_i = w'_i$, and $\varphi_{\pi,u} = \varphi_{\pi',u'}$, this means $\pi(w_i) = \pi'(w_i)$.
	\item[Case 3: $w \not\in P_\pi(u)$ and $w \in P_{\pi'}(u')$.]
	We show that it is essentially impossible to satisfy the property's condition $\pi(w) < \pi(w_j)$ in this case, implying that the property holds automatically.	First, observe that $w = w'_i$ should hold for some $i > j$, or otherwise $w\in P_\pi(u)$ by definition of the branching vertex $w_j$. By construction of $\varphi_{\pi', u'}$, this implies $\varphi_{\pi',u'}(w) \geq \pi'(w'_j) = \pi'(w_j)$. Combined with assumption $\pi'(w_j) \geq \pi(w_j)$, we get $\varphi_{\pi', u'}(w) \geq \pi(w_j)$. Moreover, since $w \not\in P_\pi(u)$, we get $\varphi_{\pi, u}(w) = \pi(w)$. Thus, to have $\varphi_{\pi, u} = \varphi_{\pi', u'}$, it should hold that $\pi(w) = \varphi_{\pi', u'}(w)$. Since we just showed $\varphi_{\pi', u'}(w) \geq \pi(w_j)$, this would imply $\pi(w) \geq \pi(w_j)$ which as outlined at the start of this case, is sufficient for our purpose.
\end{description}
The cases above clearly cover all possibilities; thus the proof is complete.
\end{proof}

\begin{proof}[Proof of Claim~\ref{cl:propofequalperms} part~\ref{prop:piuandpipupareeq}.] Recall that $w_1 = w'_1 = v$. We have $\varphi_{\pi,u}(v)=\pi(w_k)$ and $\varphi_{\pi',u'}(v)=\pi(w'_{k'})$ simply by construction of these permutations. Therefore, to have $\varphi_{\pi,u}(v) = \varphi_{\pi',u'}(v)$, we should have $\pi(w_k) = \pi'(w'_{k'})$.	
\end{proof}

\begin{proof}[Proof of Claim~\ref{cl:propofequalperms} part~\ref{prop:wjp1existsandetc}.] Suppose for contradiction that $k \leq j$, i.e., vertex $w_{j+1}$ does not exist. Since $w_j$ is the branching vertex, it has to belong to $P_\pi(u)$ by definition, thus, $k \geq j$. Combined with $k \leq j$, the only possibility would be $k=j$. By Claim~\ref{cl:propofequalperms} part~\ref{prop:piuandpipupareeq}, we have $\pi(w_k) = \pi'(w'_{k'})$ and since $j = k$, we get 
\begin{equation}\label{eq:21430}
	\pi(w_j) = \pi'(w'_{k'}).
\end{equation}
	On the other hand, by definition of $P_{\pi'}(u')$, we have
\begin{equation}\label{eq:123804}
	\pi'(w'_{k'}) > \pi'(w'_{k'-1}) \ldots > \pi'(w'_1).
\end{equation}
	Moreover, recall from the claim's assumption that $\pi'(w_j) \geq \pi(w_j)$. Combining this with (\ref{eq:21430}) and (\ref{eq:123804}), the only option is if $j = k'$. To see this, observe that $j \leq k'$ by definition of the branching vertex; now, if $j < k'$, then from (\ref{eq:123804}) we obtain $\pi'(w'_{k'}) > \pi'(w_j)$ which due to (\ref{eq:21430}) would imply $\pi(w_j) > \pi'(w_j)$ contradicting the claim's assumption that $\pi(w_j) \leq \pi'(w_j)$; thus, $j = k'$. Recall that we also assumed $j = k$ at the beginning of the proof, therefore $j = k = k'$. This implies by definition of the branching vertex that $w_i = w'_i$ for any $i \in [k]$ (or equivalently $[k']$), i.e., the two paths $P_{\pi}(u)$ and $P_{\pi'}(u')$ are exactly the same. Moreover, due to $j=k=k'$ and Claim~\ref{cl:propofequalperms} parts~\ref{prop:allbelowpiwjareequal} and \ref{prop:piuandpipupareeq}, for any vertex $w_i$ in the propagation paths, $\pi(w_i) = \pi'(w_i)$. On the other hand, for any vertex $x$ outside the two paths, we have $\pi(x) = \pi'(x)$ by Claim~\ref{cl:propofequalperms} part~\ref{prop:alloutsideequalrank}.
	Therefore, overall, the two permutations $\pi$ and $\pi'$ have to be exactly the same on all vertices, which is a contradiction with the claim's assumption that $\pi$ and $\pi'$ are different. Therefore, our initial assumption that $k \leq j$ cannot hold and vertex $w_{j+1}$ should exist.
	
Finally, by construction of $\varphi_{\pi,u}$, we have $\varphi_{\pi,u}(w_{j+1}) = \pi(w_j)$. Now, since $w_{j+1} \not\in P_{\pi'}(u')$ (otherwise $w_{j+1}$ would be be the branching vertex instead of $w_j$), we have $\varphi_{\pi',u'}(w_{j+1})=\pi'(w_{j+1})$. From $\varphi_{\pi,u} = \varphi_{\pi',u'}$, we get $\varphi_{\pi,u}(w_{j+1}) = \varphi_{\pi',u'}(w_{j+1})$. Combining these three equalities, we get $\pi(w_j)=\pi'(w_{j+1})$ as desired.
\end{proof}

\begin{proof}[Proof of Claim~\ref{cl:propofequalperms} part~\ref{prop:wjp1inmis}.] Suppose for the sake of contradiction that $w_{i+1} \not\in \lfmis{G, \pi'}$ and let $x := \eliminator{G}{\pi'}{w_{j+1}}$ be the eliminator of $w_{j+1}$ in \lfmis{G, \pi'}. Since $w_{j+1} \not\in \lfmis{G, \pi'}$, it holds that $\pi'(x) < \pi'(w_{i+1})$. Moreover, by Claim~\ref{cl:propofequalperms} part~\ref{prop:wjp1existsandetc}, $\pi'(w_{j+1}) = \pi(w_j)$; combined with inequality $\pi'(x) < \pi'(w_{j+1})$, this implies that $\pi'(x) < \pi(w_j)$. Note also that, by Claim~\ref{cl:propofequalperms} part~\ref{prop:allbelowpiwjareequal}, the two permutations $\pi$ and $\pi'$ are exactly the same on the set of vertices with rank less than $\pi(w_j)$; since $x$ is among such vertices,
\begin{equation}\label{eq:12894}
\pi(x) = \pi'(x) < \pi(w_j).
\end{equation}
Another implication of the equivalence of the two permutations on vertices with rank less than $\pi(w_j)$ is that since $x \in \lfmis{G, \pi'}$ (which holds since $x$ is the eliminator of $w_{j+1}$ in \lfmis{G, \pi'}) we also have $x \in \lfmis{G, \pi}$. This in turn, implies that $x$ is the eliminator of $w_{j+1}$ in $\lfmis{G, \pi}$ as well. On the other hand, since $w_{j+1}$ is a vertex in path $P_\pi(u)$, by definition of the propagation-paths, it should hold that $w_{j+1} \in \affectedset_\pi$. Moreover, by definition of $\affectedset_\pi$, we have $\eliminator{G}{\pi}{w_{j+1}} \not= \eliminator{G'}{\pi}{w_{j+1}}$ where $G'$ is defined as $G[V \setminus \{v\}]$. Denoting $\eliminator{G'}{\pi}{w_{j+1}}$ by $y$ and noting that $x = \eliminator{G}{\pi}{w_{j+1}}$, we get $y \not= x$. Therefore, one of the following cases should occur:
\begin{description}
	\item[Case 1: $\pi(y) < \pi(x)$.] In this case, the fact that $x$ is the eliminator of $w_{j+1}$ in \lfmis{G, \pi} even though $\pi(y) < \pi(x)$ means $y \not \in \lfmis{G, \pi}$. On the other hand, $y \in \lfmis{G', \pi}$ since $y = \eliminator{G'}{\pi}{w_{j+1}}$, therefore $y \in \flippedset_\pi$ by definition. However, this contradicts $w_j = \parent{\pi}{w_{j+1}}$ since by (\ref{eq:12894}), $\pi(y) < \pi(x) < \pi(w_j)$ , thus, $y$ should be the parent of $w_{j+1}$ instead of $w_j$.
	\item[Case 2: $\pi(y) > \pi(x)$.] Similarly, in this case, the fact that $x$ is not the eliminator of $w_{j+1}$ in \lfmis{G', \pi} even though $\pi(x) < \pi(y)$ implies that $x \not\in \lfmis{G', \pi}$. This means that $x \in \flippedset_\pi$ and again, since $\pi(x) < \pi(w_j)$, $x$ has to be the parent of $w_{j+1}$ instead of $w_j$.
\end{description}
To wrap up, $w_{j+1} \not\in \lfmis{G, \pi'}$ leads to a contradiction, thus $w_{j+1} \in \lfmis{G, \pi'}$.
\end{proof}

\subsection{Unlikely Permutations: Proof of Lemma~\ref{lem:unlikeliesareunlikely}}\label{sec:unlikely}

The LFMIS over a permutation $\pi$ can be parallelized in the following way. In each  round, each vertex that holds the minimum rank among its neighbors joins the MIS and then is removed from the graph along with its neighbors (note that this, in parallel, happens for several vertices in each round). Fischer and Noever \cite{DBLP:conf/soda/FischerN18}, building on an earlier approach of Blelloch \etal{}~\cite{DBLP:conf/spaa/BlellochFS12}, showed that if permutation $\pi$ is chosen randomly, with probability at least $1-n^{-2}$, it takes $O(\log n)$ rounds until the graph becomes empty.\footnote{We note that the success probability of these works is actually $1-n^{-c}$ for any desirable constant $c>1$ affecting the hidden constants in the round-complexity. For our purpose, $c=2$ is sufficient.} This result as a black-box does not prove Lemma~\ref{lem:unlikeliesareunlikely}. However, to prove this upper-bound on round-complexity, they indeed upper bound the maximum size of {\em dependency-paths} which are structures that are very close to  propagation-paths:

\begin{definition}[{\cite[Definition~2.1]{DBLP:conf/soda/FischerN18}}]
	A path $w_1, w_2, \ldots, w_k$ in the graph is a dependency-path according to permutation $\pi$, if for any odd $i \in [k]$, vertex $w_i$ is in \lfmis{G, \pi} and for any even $i \in [k]$, $w_i \not\in \lfmis{G, \pi}$ and $w_{i-1} = \eliminator{G}{\pi}{w_i}$.
\end{definition}

Recall that indeed, if $u_1, \ldots, u_k$ is a propagation-path, then for every $i \in [k-1]$, $u_i = \eliminator{G}{\pi}{u_{i+1}}$ by definition. Moreover, by Claim~\ref{cl:alternateinMIS}, except for the last vertex in the propagation-path, the odd vertices are in the MIS and the even vertices are not. Therefore:

\begin{observation}\label{obs:prop-dep}
	If there exists a propagation-path of size $\ell$ in the graph, then its first $\ell-1$ vertices form a dependency-path.
\end{observation}

Fischer and Noever~\cite{DBLP:conf/soda/FischerN18} prove that with probability $1-n^{-2}$, every dependency-path has size $O(\log n)$ if $\pi$ is chosen at random. Therefore, from Observation~\ref{obs:prop-dep}, we get that the probability of having a propagation-path with size $\beta \log n$, if $\beta$ is a large enough constant, is at most $n^{-2}$, which completes the proof of Lemma~\ref{lem:unlikeliesareunlikely}.

\section{Fully Dynamic MIS: Putting Everything Together}\label{sec:wrapup}
\subsection{The (Concrete, Non-Parametrized) Running Time}
In this section, we show how combining Lemma~\ref{lem:updatealg} with Lemma~\ref{lem:edgeA} proves the main claim of this paper that MIS can be maintained in  polylogarithmic update-time.

\restate{Theorem~\ref{thm:main}}{\thmmain{}}

\begin{proof}[Proof of Theorem~\ref{thm:main}]
	Consider insertion or deletion of an edge $e=(a, b)$. As before, we use $\lambda$ to denote random variable $\min\{ \pi(a), \pi(b) \}$ and use $\affectedset$ to denote the set of vertices whose eliminators change as a result of this edge update. By Lemma~\ref{lem:updatealg}, we have
	\begin{flalign}
		\nonumber \E[\text{update-time for an edge $e=(a, b)$}] &= \E \left[O\left(|\affectedset| \cdot \log \Delta \cdot \min\left\{\lambda^{-1} \cdot \log n, \Delta \right\}\right)\right]\\
		\nonumber &= O(\log \Delta) \cdot \E \left[|\affectedset| \cdot \min\left\{\lambda^{-1} \cdot \log n, \Delta \right\}\right]\\
		&= O(\log \Delta \cdot \log n) \cdot \E \left[\min\left\{\lambda^{-1} \cdot \log n, \Delta \right\}\right].\label{eq:3918247}
	\end{flalign}
	The third equation follows from $\E[|\affectedset| \mid \lambda] \leq O(\log n)$ which was proved in Lemma~\ref{lem:edgeA}, combined with the fact that if for two possibly dependent random variables $y_1$ and $y_2$, $\E[y_1 | y_2] \leq \beta$, then $\E[y_1 \cdot y_2] \leq \beta \E[y_2]$. To bound the random variable inside the expectation, suppose we partition the $[0, 1]$ interval into $\Delta$ sub-intervals $I_1, \ldots, I_\Delta$ where $I_i = [\frac{i-1}{\Delta}, \frac{i}{\Delta}]$ for any $i \in [\Delta]$. Note that if $\lambda \in I_i$ then at least one of $\pi(a)$ and $\pi(b)$ is in $I_i$. Therefore, a simple union bound implies that $\Pr[\lambda \in I_i] \leq \Pr[\pi(a) \in I_i] + \Pr[\pi(b) \in I_i] = 2/\Delta$. We, thus, have:
		\begin{flalign*}
		\E \left[\min\left\{\lambda^{-1} \cdot \log n, \Delta \right\}\right] &= \sum_{i=1}^{\Delta} \Pr[\lambda \in I_i] \cdot \E \left[\min\left\{\lambda^{-1} \cdot \log n, \Delta \right\} \mid \lambda \in I_i \right]\\
		&\leq \sum_{i=1}^{\Delta} \frac{2}{\Delta} \Big( \min\Big\{\frac{\Delta}{i-1} \log n, \Delta \Big\} \Big) = O(\log n) \sum_{i=1}^{\Delta} \frac{1}{i} = O(\log \Delta \cdot \log n).
		\end{flalign*}
		Replacing this into (\ref{eq:3918247}) suffices to bound the expected update-time by $O(\log^2 \Delta \cdot \log^2 n)$. Furthermore, as mentioned before in Section~\ref{sec:highlevel}, we already know from \cite[Theorem~1]{Censor-HillelHK16} that the expected adjustment complexity of random order LFMIS is $O(1)$, completing the proof.
\end{proof}

\subsection{Deferred Proofs}\label{sec:proofsMIS}

We start by proving Observation~\ref{obs:inAiffalowerchanges} which is crucial for the algorithm's correctness.

\restate{Observation~\ref{obs:inAiffalowerchanges}}{\obsinAiff{}}

\begin{proof}[Proof of Observation~\ref{obs:inAiffalowerchanges} part 1]
Let $U$ denote the set of vertices $v$ in $V$ with $\pi(v) < \pi(a)$. Observe that the two induced subgraphs $G_t[U]$ and $G_{t-1}[U]$ are identical since the only difference between $G_t$ and $G_{t-1}$ is insertion/deletion of edge $e=(a,b)$ whose endpoints both have rank at least $\pi(a)$ (recall that $\pi(a) < \pi(b)$) and thus neither belongs to $U$. Since the MIS is constructed greedily on lower rank vertices first, the set of MIS vertices in $G_t[U]$ and $G_{t-1}[U]$ according to $\pi$ are exactly the same.  Let $I_U$ denote these MIS nodes. Note that any vertex $v$ with $k_{t-1}(v) < \pi(a)$ should have a neighbor in $I_U$. Since both end-points of edge $e$ are in $V \setminus U$, the set of neighbors of $I_U$ in both graphs $G_t$ and $G_{t-1}$ are also identical. Therefore for each vertex $v$ with $k_{t-1}(v) < \pi(a)$, we have $k_t(v) = k_{t-1}(v)$ and thus $v$ cannot be in $\affectedset$ by definition.  By a similar argument, for any vertex $v$ with $k_t(v) < \pi(a)$ we also have $k_{t-1}(v)  = k_t(v)$ and thus $v \not\in \affectedset{}$.
\end{proof}

\begin{proof}[Proof of Observation~\ref{obs:inAiffalowerchanges} part 2]
	The assumption $v \in \affectedset$ implies that the eliminator of $v$ has changed after the update. Let $w$ be the eliminator of $v$ before the update. If the MIS-status of no neighbor $u$ of $v$ with $\pi(u) \leq \pi(w)$ changes, since $v \not= b$ and the set of neighbors of $v$ are the same before and after the update, then $w$ remains to be the eliminator of $v$. Therefore, to have $v \in \affectedset$, the MIS-status of at least one of $v$'s neighbors changes and this vertex is in $\flippedset$ by definition.
\end{proof}

We first prove the correctness of each of the subroutines and then that of the overall algorithm. These subroutines are proven to be correct by the end of any iteration $i$ conditioned on the assumption that Invariants~\ref{inv:S1}-\ref{inv:N} (or a subset of them) hold at the start of iteration $i$. We later inductively prove that these invariants hold and that indeed the whole algorithm is correct.

\begin{claim}\label{cl:isaffectedcorrectnessandtime}
	By the end of any iteration $i$, subroutine $\isaffected{v}$ correctly decides whether the lowest-rank vertex $v \in \S$ is in set $\affectedset$ in time $O(|\P{v}|)$ given that Invariants~\ref{inv:S1}-\ref{inv:N} hold by the start of iteration $i$.
\end{claim}
\begin{proof}
	The algorithm clearly takes $O(|\P{v}|)$ time since it only iterates over the vertices in $\P{v}$ to decide on the output. In what follows, we prove its correctness. As in the algorithm's description, we consider two cases where $v = b$ and $v \not= b$ individually.
	
	\smparagraph{Case 1 $v = b$.} In this case, the algorithm decides $b \in \affectedset$ if and only if $m(a) = 1$ and $k(b) \geq \pi(a)$. We show that this is indeed correct.
	
	\textit{The if part.} We show that if  $m(a) = 1$ and $k(b) \geq \pi(a)$, then $b\in \affectedset$. Observe from Invariant~\ref{inv:S1} that at this point in the algorithm, we have $k(b) = k_{t-1}(b)$. Therefore, the $k(b) \geq \pi(a)$ assumption implies $k_{t-1}(b) \geq \pi(a)$. Moreover, the MIS-status of vertex $a$ cannot change as it is the lower-rank vertex of the updated edge, thus, it holds that $m_t(a) = m_{t-1}(a)$ and consequently $m(a)=1$ implies $a \in \lfmis{G_{t-1}, \pi}$. Combining these, the eliminator of $b$ has to be $a$ iff there is an edge between $a$ and $b$. Therefore, updating edge $e=(a,b)$ definitely changes $b$'s eliminator and thus $b \in \affectedset$.
	
	\textit{The only if part.} Suppose that one of the conditions do not hold, we show $ b \not\in \affectedset$. First, if $k < \pi(a)$, then by Observation~\ref{obs:inAiffalowerchanges} part 1, $b \not\in \affectedset$ as desired. Moreover, if $m(a)=0$, as before, we should have $m_{t-1}(a) = m_t(a) = 0$ since $a$ is the lower-rank vertex of the update. As a result, insertion or deletion of $e$ cannot have an effect on the eliminator of $b$ and thus $b \not\in \affectedset$.
	
	\smparagraph{Case 2 $v \not= b$.} In this case, the eliminator of $v$ changes if and only if at least one of the following conditions hold: (1) the eliminator of $v$ in time $t-1$ leaves the MIS, (2) at least a vertex $u$ adjacent to $v$ with $\pi(u) < k(v)$ joins the MIS. If none of these conditions hold, then $\eliminator{G_{t-1}}{\pi}{v}$ remains to be the smallest-rank vertex in $\{b\} \cup N(b)$ that is in the MIS after the update; therefore  by definition of eliminator, $k_{t-1}(v)=k_t(v)$ and thus $v \not\in \affectedset$. 
	
	Our algorithm precisely checks these conditions. For condition (1), if the eliminator $u := \eliminator{G_{t-1}}{\pi}{v}$ leaves the MIS after the update, it should by definition belong to $\flippedset$. Note that by invariant~\ref{inv:S2}, \P{v} exactly contains the neighbors $w$ of $v$ with $w \in \flippedset$ and $\pi(w) < \pi(v)$. Therefore if $u \in \P{v}$, then condition (1) holds and $v \in \affectedset$. Our algorithm also checks condition (2) by finding the lowest-rank vertex $w$ in \P{v} with $m(w)=1$ and then comparing its rank with $k_{t-1}(v)$.
\end{proof}
%
%

%

\begin{claim}\label{cl:correctfindrelevantneighbors}
	At any iteration $i$, with probability at least $1-n^{-(c+1)}$, set $\mathcal{H}_v$ has size $O(\min\{\Delta, \frac{\log n}{\pi(a)}\})$. Furthermore, subroutine $\findrelneighbors{v, \pi(a)}$ correctly finds the set $\mathcal{H}_v$ in time $O(|\mathcal{H}_v| \cdot \log \Delta)$, given that Invariant \ref{inv:N} holds by the start of iteration $i$.
\end{claim}
\begin{proof}
	\smparagraph{Size of $\mathcal{H}_v$.} Observe that if $\mathcal{H}_v$ is defined, then as assured by the condition in Line~\ref{line:if} of Algorithm~\ref{alg:update}, $v \in \affectedset$ thus by Observation~\ref{obs:inAiffalowerchanges}, $k_{t-1}(v) \geq \pi(a)$. Furthermore, by definition, every vertex $u \in \mathcal{H}_v$ has $k_{t-1}(u) \geq \pi(a)$. This means that if we take LFMIS of $G_{t-1}$ induced on vertices with rank in $[0, \pi(a))$ and remove them and their neighbors from the graph, $v$ and all of its neighbors in $\mathcal{H}_v$ will survive. Recall that the adversary is oblivious and the graph $G_{t-1}$ and random permutation $\pi$ are chosen independently. Therefore, applying Proposition~\ref{prop:sparseMIS} on graph $G_{t-1}$ with $p = \pi(a)$ bounds $|\mathcal{H}_v|$ by $O(\pi(a)^{-1} \log n)$ w.h.p. Moreover, clearly $|\mathcal{H}_v| \leq \Delta$ since they are neighbors of $v$, concluding the bound on the size of $\mathcal{H}_v$.
		
	\smparagraph{Correctness.} The assumption that Invariant~\ref{inv:N} holds implies that $N^-(v) = N^-_{t-1}(v)$ and $N^+(v)=N^+_{t-1}(v)$. Therefore, \findrelneighbors{v, \pi(a)} correctly finds $\mathcal{H}_v$.
	
	\smparagraph{Running time.} Since the vertices  $u \in N^-(v)$ are indexed by $k_{t-1}(u)$ and the algorithm iterates only over the neighbors $u$ of $v$ in this set with $k_{t-1}(u) \geq \pi(a)$, the running time of this part is $O(|\mathcal{H}_v| \log \Delta)$ where the $\log \Delta$ factor comes from searching in this BST which has size $\Delta$ at most. However, note that the algorithm iterates over all vertices in $N^+(v)$ since it is not indexed by $k_{t-1}(.)$. Therefore, we have to prove $|N^+(v)|$ cannot be larger than $|\mathcal{H}_v|$. We know from Invariant~\ref{inv:N} that for any vertex $u \in N^+(v)$, we have $k_{t-1}(u) \geq k_{t-1}(v)$. Moreover, since $v \in \affectedset$, by Observation~\ref{obs:inAiffalowerchanges}, $k_{t-1}(v) \geq \pi(a)$. Combining the two, we get that $k_{t-1}(u) \geq \pi(a)$. This means that every vertex $u \in N^+(v)$ that is still a neighbor of $v$ after the update, should be in set $|\mathcal{H}_v|$. Since at most one edge is removed from the graph at time $t$, we have $|N^+(v)| \leq |\mathcal{H}_v|+1$, completing the proof.
\end{proof}

\begin{claim}\label{claim:updatek}
Let $v$ be the lowest-rank vertex at the start of an arbitrary iteration. Subroutine \updateeliminator{v, \mathcal{H}_v} correctly updates $k(v)$ and $m(v)$ of vertex $v$ in time $O(|\mathcal{H}_v|)$ assuming that Invariant~\ref{inv:S1} holds by this iteration.
\end{claim}
\begin{proof}
	It is clear that the algorithm takes $O(|\mathcal{H}_v|)$ time, here we prove its correctness. Note that at the time of using subroutine \updateeliminator{v, \mathcal{H}_v}, we know $v \in \affectedset$. Therefore, from Observation~\ref{obs:inAiffalowerchanges} part 1, we know $k_t(v) \geq \pi(a)$ and $k_{t-1}(v) \geq \pi(a)$. We consider the two cases where $m_t(v)=1$ and $m_t(v)=0$ differently.
	
	Suppose that $m_t(v)=0$ and let $w$ be the eliminator of $v$ after the update, i.e., $\pi(w) = k_t(v)$ (note that since $m_t(v)=0$, $w \not= v$). We first show $w \in \mathcal{H}_v$ by proving that $k_{t-1}(w) \geq \pi(a)$. Suppose for contradiction that $k_{t-1}(w) < \pi(a)$. Then by Observation~\ref{obs:inAiffalowerchanges} part 1, $w \not\in \affectedset$ and consequently $w \not\in \flippedset$ since $\flippedset \subseteq \affectedset$. Since $w$ is the eliminator of $v$ in $G_t$, we have $m_t(w)=1$. Moreover, for $w \not\in \flippedset$, we also get $m_{t-1}(w)=1$ which, by definition, means $w$ has to be its own eliminator in $G_{t-1}$ and thus $k_{t-1}(w) = \pi(w)$. Combined with $k_{t-1}(w) < \pi(a)$, this would mean $\pi(w) < \pi(a)$. This, however, contradicts $k_{t-1}(v) \geq \pi(a)$ since $v$ has a neighbor $w$ in MIS of $G_{t-1}$ with rank smaller than $\pi(a)$ and thus it should hold that $k_{t-1}(v) < \pi(a)$. This contradiction implies that indeed $k_{t-1}(w) \geq \pi(a)$ and thus $w \in \mathcal{H}_v$. Furthermore, in this case, since $\pi(w) < \pi(v)$, by Invariant~\ref{inv:S1}, $m(w) = m_{t}(w) = 1$ and indeed the lowest-rank vertex $u$ in $\mathcal{H}_v$ with $m(u)=1$ should be vertex $w$ and the algorithm is correct.
	
	On the other hand, if $m_t(v)=1$, then no lower-rank neighbor of $v$ should be in the MIS. In this case, once we scan the set $\mathcal{H}_v$, we will not find any vertex $u$ with a lower-rank than $\pi(v)$ and $m(u)=1$, thus we correctly decide that $v$ is in the MIS and update $m(v)$ and $k(v)$ correctly.
\end{proof}

\begin{claim}\label{cl:adjlistscorrect}
	Subroutine \updateadjlists{} correctly updates the adjacency lists and with probability at least $1-n^{-c}$, takes $O(|\affectedset|\cdot \min\{\Delta, \frac{\log n}{\pi(a)}\} \cdot \log \Delta)$ time given that for any vertex $v$, $k(v) = k_t(v)$.
\end{claim}
\begin{proof}
The only edge update is between vertices $a$ and $b$ and the algorithm first accordingly addresses this change by updating $N^+(a)$,  $N^-(a)$,  $N^+(b)$, and  $N^-(b)$.  For the rest of the vertices, we do not have an edge update but the changes to the adjacency lists are resulted by the changes to the eliminators. For a vertex $v$, these changes are limited to moving its neighbors between $N^+(v)$ and $N^-(v)$ or possibly re-indexing its neighbors in $N^-(v)$ whose eliminator has changed.

We say an edge $(v, u)$ causes an update iff position of $u$ and $v$ or their indexing in each others' adjacency lists ($N^+$ or $N^-$) needs to be updated. Let  $T$ denote the set of these edges. Note that by definition of $N^+$ and $N^-$, if $u \notin \A$ and $v \notin \A$, then  $(v, u)\notin T$. This means that at least one end-point of any edge in $T$ is in $\A$.

\=Assume w.l.o.g. that for edge $(v, u)\in T$, we have $v\in \A$. We claim that $u\in \mathcal{H}_v$ should hold. To show this, we assume that $u\notin \mathcal{H}_v$ and obtain a contradiction. Recall that we have $u\notin \mathcal{H}_v$ iff $k(u) < \pi(a)$. By Observation~\ref{obs:inAiffalowerchanges} part 1, this would imply  $k_{t-1}(u) < k_{t-1}(v)$, $k_{t}(u) < k_{t}(v)$, and $u \notin \A$. Because of $k_{t-1}(u) < k_{t-1}(v)$ and $k_{t}(u) < k_{t}(v)$, the position of vertices $u$ and $v$ in each others adjacency lists remains unchanged. That is, we have $v \in N^+(u)$, $v \notin N^-(u)$, $u\in N^-(v)$, and $u\notin N^-(v)$ at both times $t$ and $t-1$.
 Moreover, since $u \not\in \affectedset$, we have $k_{t-1}(u) = k_t(u)$ and thus $u$ is already correctly indexed in $N^-(v)$. This is, however, a contradiction since position of $u$ and $v$ and their indexing in each others' adjacency lists is already updated and as a result $(v, u)\notin T$. Therefore, it should indeed hold that $u \in \mathcal{H}_v$.
 
 In subroutine \updateadjlists{}, for any vertex $v\in \A$ we go over its neighbors $u\in \mathcal{H}_v$ and determine the membership of vertex $v$ in adjacency lists of vertex $u$ and vice versa. To do so, by definition of $N^+$ and  $N^-$ we only need values of $k_{t}(v)$ and $k_{t}(u)$ which are assumed to be updated (in the statement of the claim). We then update $N^-(v)$, $N^+(v)$, $N^-(u)$ and $N^+(u)$ accordingly; thus the algorithm correctly updates the adjacency lists.
 
 To analyze the running time, using Claim~\ref{cl:correctfindrelevantneighbors}, we know that for any vertex $v\in \A$, set $\mathcal{H}_v$ has size  $O(\min\{\Delta, \frac{\log n}{\pi(a)}\})$ with probability at least $1-n^{-(c+1)}$. Also, each update takes  $O(\log \Delta)$ time since it consists of at most four insertions and deletions in adjacency lists which are stored as BSTs. Overall, this means that the running time can be bounded by $O(|\affectedset|\cdot \min\{\Delta, \frac{\log n}{\pi(a)}\} \cdot \log \Delta)$ with probability at least $1-n^{-c}$.
\end{proof}

\begin{claim}\label{claim:inv1-inv2}
	If Invariant~\ref{inv:S1} holds by some iteration $i$, then Invariant~\ref{inv:S2} also holds by iteration $i$.
\end{claim}
\begin{proof}
	Let $u$ be any vertex adjacent to $v$ with $\pi(u) < \pi(v)$ and $u\in \flippedset$. In other words, any vertex that should be in set \P{v} for the Invariant~\ref{inv:S2} to hold. Assuming that  Invariant~\ref{inv:S1} holds, we know that $m(u) = m_{t}(u)$ and $m(u) \neq m_{t-1}(u)$. Observe that in the algorithm, updating $m(u)$ only happens in subroutine $\updateeliminator{v, \mathcal{H}_v}$ which is followed by adding $u$ to set $\P{.}$ of any vertex in set $\mathcal{H}_u$ if $u$ is flipped. Set $\mathcal{H}_u$ by definition includes vertex $v$ since $k(v)\geq \pi(a)$ and $\pi(v)> \pi(u)$. This proves that set $\P{v}$ satisfies Invariant~\ref{inv:S2}.
\end{proof}

\begin{claim} \label{claim:inv1}
Let $v$ be the lowest-rank vertex in $\S$ in an arbitrary iteration $i$ of the algorithm. Assuming that Invariant~\ref{inv:S1} holds at the start of iteration $i$ we have:
\begin{enumerate}
	\item If $\S = \emptyset$ at the end of iteration $i$, for any vertex $u\in V$, $m(u) = m_{t}(u)$ and $k(u) = k_t(u)$.
	\item If  $\S \neq \emptyset$ at the end of iteration $i$, then Invariant~\ref{inv:S1} holds at the start of iteration $i+1$ as well.   
\end{enumerate}

\end{claim}

\begin{proof}
Let $\S'$ denote set $\S$ at the end of iteration $i$ and let $v'$ be the lowest-rank vertex in that. Throughout the proof, by $\S$ we mean set $\S$ at the start of iteration $i$ and we use $v$ to refer to its lowest-rank vertex. Let us first review 
Invariant~\ref{inv:S1}. It states that for any vertex $u$,  if $\pi(u)< \pi(v)$ then  $k(u)=k_{t-1}(u)$, and $m(u)=m_{t-1}(u)$ hold and otherwise we have $k(u)=k_{t-1}(u)$, and $m(u)=m_{t-1}(u)$. We first show that  $k(v)$ and $m(v)$ are updated at the end of iteration $i$. By claim~\ref{claim:inv1-inv2}, we know that Invariant~\ref{inv:S2} holds at the start of iteration $i$ and by Claim~\ref{cl:isaffectedcorrectnessandtime}, we know that holding Invariant~\ref{inv:S2} means that subroutine $\isaffected{v}$ correctly detects if $v\in \A$ or not. Moreover, by Claim~\ref{claim:updatek} if $v\in \A$, in the next step, algorithm correctly updates $k(v)$ and $m(v)$. At this point of the algorithm, we know that for any vertex $u$ with $\pi(u) \leq \pi(v)$, we have  $k(u)= k_{t}(u)$ and $m(u)= m_{t}(u)$. 

Now, let $u$ be the vertex with the lowest-rank among the vertices in $\A$ whose rank is greater than $\pi(v)$. To complete the proof it suffices if we show that if such a vertex exists, then $u\in \S'$. This means that if $S' = \emptyset$, then for any vertex $u\in V$, we have $m(u) = m_{t}(u)$ and $k(u) = k_t(u)$. Moreover, for the case of $S' \neq \emptyset$, it results that for any vertex $u$, with $\pi(u) < \pi(v')$ we have $m(u) = m_{t}(u)$ and $k(u) = k_t(u)$ or in the other words that Invariant~\ref{inv:S1} holds at the start of iteration $i+1$. We use proof by contradiction by assuming that there exists a vertex $u$ in set $\A$ but not in $S'$ such that for any vertex $u'$ with $\pi(u') < \pi(u)$ we have  $m(u')= m_t(v')$, and $k(u') = k_t(u')$. By Observation~\ref{obs:everyvertexhasneighborinF}, any vertex in $\A$ has a neighbor in $\flippedset$ with a lower rank. Let $u'$ be such a neighbor of $u$. By the assumption that all neighbors of $u$ with a lower rank has updated $m(.)$, we have $m(u') \neq m_{t-1}(u')$.  Observe that in the algorithm, updating $m(u')$ only happens in subroutine $\updateeliminator{v, \mathcal{H}_{u'}}$ which is followed by adding vertices in $\mathcal{H}_u$ to $\S$. 
Set $\mathcal{H}_u$, by definition, includes vertex $u$ since $k(u)\geq \pi(a)$ (otherwise by Observation~\ref{obs:inAiffalowerchanges}, $u \notin \A$ ) and $\pi(u) > \pi(u')$. Thus, we obtain a contradiction and the proof is completed.
\end{proof}

\begin{claim}\label{cl:invshold}
	Invariants~\ref{inv:S1}, \ref{inv:S2}, and \ref{inv:N} hold throughout the algorithm with probability 1.
\end{claim}
 
\begin{proof}
First, observe that Invariant~\ref{inv:N} holds since Line~\ref{line:updateN} is the only part of the algorithm that we modify the adjacency lists. Moreover, by Claim~\ref{claim:inv1-inv2}, the correctness of Invariant~\ref{inv:S2} results from  Invariant~\ref{inv:S2}. Thus, we only need to show that 	Invariants~\ref{inv:S1} holds throughout the algorithm. We do so using induction. As the base case, in the first iteration of the algorithm we have $S=\{b\}$ (or $S=\emptyset $ which does not need a proof). We need to show that for any vertex $u$ if $\pi(u)< \pi(b)$ we have $k(u)=k_t(u)$, and $m(u)=m_t(u)$ and  if $\pi(u)> \pi(b)$ we have $k(u)=k_{t-1}(u)$, and $m(u)=m_{t-1}(u)$. Before the start of this iteration we have not changed $k(u)$ and $m(u)$ of any vertex $u$ thus for all of them $k(u) = k_{t-1}(u)$ and $m(u)=m_{t-1}(u)$. Moreover, by Observation~\ref{obs:everyvertexhasneighborinF}, updating edge $e$ does not affect a vertex $u$ with $\pi(u) < \pi(b)$ which means that for any such vertex we have $k_t(u)= k_{t-1}(u)$. Therefore,  we conclude that Invariants~\ref{inv:S1}  holds for the base case. This completes the proof since 
		the induction step is a direct result of Claim~\ref{claim:inv1}.
	\end{proof}
	
We continue with a simple observation and then turn to formally prove the running time.

\begin{observation}\label{obs:increasing}
Let $v_i$ and $v_j$ respectively denote the lowest-rank vertices of \S{} in two arbitrary iterations $i$ and $j$ of Algorithm~\ref{alg:update}. If $i<j$ then $\pi(v_i) < \pi(v_j)$. 
\end{observation}

\begin{proof}
We show that this claim holds for $j = i+1$ which can be inductively used to generalize it to any arbitrary $i$ and $j$. Let $\S_{i}$ and $\S_{i+1}$ respectively denote set $\S$ at the beginning of iteration $i$ and set $\S$ at the beginning of iteration  $i+1$. We know that $v_{i+1}$ is either inserted to $\S$ in iteration $i$ or that it is in set $\S_{i}$. Observe that any vertex added to $\S$ in the $i$-th iteration has rank lower than $\pi(v_i)$ and that $v_i$ is the lowest-rank vertex in $\S_{i}$. As a result $\pi(v_i) < \pi(v_{i+1})$.
\end{proof}

\begin{claim} \label{claim:runningtime}
	With probability at least $1-n^{-c}$, the total running time of the algorithm until the set $\S$ becomes empty is at most $O(|\affectedset| \cdot \log \Delta \cdot \min\{\frac{\log n}{\pi(a)}, \Delta \})$.	
\end{claim}
\begin{proof}
To prove this claim, we first show that $|\S|$ and $\sum_{v\in \S} |\P{v}|$ are both  $O(|\affectedset| \cdot \log \Delta \cdot \min\{\frac{\log n}{\pi(a)}, \Delta \})$ with probability at least $1-n^{-c}$. Observe that in the algorithm, we only add vertices to these sets in Line~\ref{line:addthings}. Moreover, by Observation~\ref{obs:increasing}, each vertex is removed from $\S$ at most once. Thus, the algorithm runs this  line for any vertex $v\in \A$ and any vertex $u$ in its $\mathcal{H}_v$ only once. Therefore, by Claim~\ref{cl:correctfindrelevantneighbors}, the number of times the algorithm adds a vertex to these sets adds up to $O(|\affectedset| \cdot \log \Delta \cdot \min\{\frac{\log n}{\pi(a)}, \Delta \})$ with probability at least $1-n^{-c}$.
Note that $|\S|$ is equal to the number of iterations  in the algorithm and  $\sum_{v\in \S} |\P{v}|$ is the overall time that the subroutine $\isaffected{v}$ takes over all iterations. Moreover, for any vertex   $v\in \A$ we run Lines~\ref{line:if}-\ref{line:addthings} of the algorithm which by Claim~\ref{claim:updatek} and Claim~\ref{cl:correctfindrelevantneighbors} take $O(\log \Delta \cdot \min\{\frac{\log n}{\pi(a)}, \Delta \})$ time. To sum up, the total running time of the algorithm until the set $\S$ becomes empty is $O(|A|\cdot \log \Delta \cdot \min\{\frac{\log n}{\pi(a)}, \Delta \})$ with probability at least $1-n^{-c}$.
\end{proof}

We are now ready to prove Lemma~\ref{lem:updatealg}.

\restate{Lemma~\ref{lem:updatealg}}{\updatealg{}}
%
\begin{proof}
By Claim~\ref{claim:runningtime}, with probability at least $1-n^{-c}$ it takes $O(|\affectedset| \cdot \log \Delta \cdot \min\{\frac{\log n}{\pi(a)}, \Delta \})$ time until set $\S$ becomes empty. We further show that when this happens we have $m(v) = m_t(v)$ and $k(v) = k_t(v)$. This is a direct result of Claim~\ref{cl:invshold} and Claim~\ref{claim:inv1}. The former stated that Invariant~\ref{inv:S1} holds throughout the algorithm and the latter states that if Invariant~\ref{inv:S1} holds in the last iteration of the algorithm, then for any vertex $u$ we have $m(v) = m_t(v)$ and $k(v) = k_t(v)$.	 Moreover, using Claim~\ref{cl:adjlistscorrect} we know that subroutine \updateadjlists{} correctly updates the adjacency lists and with probability at least $1-n^{-c}$, it takes $O(|\affectedset|  \cdot \log \Delta \cdot \min\{\Delta, \frac{\log n}{\pi(a)}\})$ time given that for any vertex $v$ we have $k(v) = k_t(v)$. This completes the proof and we obtain that  with probability at least $1-n^{-c}$, Algorithm~\ref{alg:update} correctly updates all the data structures in $O(|\affectedset| \cdot \log \Delta \cdot \min\{\frac{\log n}{\pi(a)}, \Delta \})$ time.
\end{proof}



%
%
%

\section{Extension to Fully Dynamic Maximal Matching}\label{sec:MM}

\smparagraph{What is Different?} It is well-known that MM of a graph can be found by first taking its line-graph and then constructing an MIS on it. Doing so, the edges in the original graph that correspond to the MIS nodes in the line-graph will form an MM. However, the line-graph may be much larger than the original graph and thus expensive to construct and maintain. Nonetheless, because of the very specific structure of line-graphs, we can indeed implement (a simpler variant of) the same algorithm for MM without going through an explicit construction of the line-graph. In what follows, we highlight the main differences between our MIS algorithm and its MM implementation.

The first difference is that for LFMM, the random ranking $\pi$ has to be drawn on the edges instead of the vertices and thus we cannot fix $\pi$ in the pre-processing step. However, this is easy to handle: We draw the rank $\pi(e) \in [0, 1]$ of any edge $e$ randomly upon its arrival.

The second difference is where the specific structure of line-graphs helps significantly. The set of edges whose MM-statuses change as a result of an edge update form a single path or a single cycle. In fact, this holds true for any arbitrary ranking $\pi$ over the edges. This is in sharp contrast with MIS, where the propagations may branch (consider a star and assume that the center leaves the MIS). This branching is precisely what complicates the proof of Theorem~\ref{thm:vertexA} for MIS. Since we do not have this problem for MM, we can directly bound the set of edges with different MM-statuses by $O(\log n)$, w.h.p., using a reduction to the parallel round complexity of random-order LFMM \cite{DBLP:conf/spaa/BlellochFS12,DBLP:conf/soda/FischerN18}. Therefore, the analog of Theorem~\ref{thm:vertexA} for MM is significantly easier to prove. It also simplifies the algorithm we use to  detect the changes to MM (compared to MIS).

The third difference is simple, but plays a crucial role in both adapting the MIS algorithm to MM and also simplifying it. Instead of storing the adjacency lists on the edges, which is the natural idea if one constructs the line-graph explicitly, we can simply store them on the vertices. In fact, because of this difference, it also turns out that for MM, we do not need to partition the adjacency lists into $N^+$ and $N^-$. That is, we can afford to keep an adjacency list $N(v)$ on each vertex $v$ including all incident edges to $v$, where each edge $e \in N(v)$ is indexed by its eliminator's rank. The main reason that this is feasible, here, is that if the eliminator of an edge $e=(u, v)$ changes, we only need to re-index $e$ in $N(u)$ and $N(v)$. However, for MIS, if the eliminator of a vertex $u$ changes, we may have to re-index $u$ in the adjacency lists of all of its neighbors.

\smparagraph{Algorithm Setup.} Suppose that we have fixed the ranking $\pi$ on the edges. As described above, we can draw $\pi(e) \in [0, 1]$ for any edge $e$ in the graph at the time of its arrival. In what follows, considering update number $t$, which can be an edge insertion or deletion, we describe how to address it and update $\lfmm{G_{t-1}, \pi}$ to $\lfmm{G_t, \pi}$ in $\polylog n$ time.

Analogous to the MIS algorithm, we define $\affectedset := \{w \mid \eliminator{G_{t}}{\pi}{w} \not= \eliminator{G_{t-1}}{\pi}{w} \}$ to be the set of edges whose eliminator changes after the update and call these the {\em affected} edges. Moreover, we define $\flippedset$ to be the set of edges whose MM-status changes after the update; we call these the {\em flipped} edges. Note that $\flippedset \subseteq \affectedset$. We first provide the following algorithm.

\newcommand{\updatealgmatching}[0]{There is an algorithm to update \lfmm{G, \pi} and the data structures required for it after insertion or deletion of any edge $f=(a, b)$ in $O\left(|\flippedset| \min\{\Delta, \frac{\log n}{\pi(f)} \} \log\Delta  \right)$ time, w.h.p.}

\begin{lemma}\label{lem:updatealgmatching}
	\updatealgmatching{}
\end{lemma}

Note a subtle difference between Lemma~\ref{lem:updatealgmatching} and the similar Lemma~\ref{lem:updatealg} we had for MIS: Here, the running time is parametrized by $|\flippedset|$ whereas in Lemma~\ref{lem:updatealg} it is parametrized by $|\affectedset|$.

We will later prove in Section~\ref{sec:analysisMM} that the running time in Lemma~\ref{lem:updatealgmatching} is actually bounded by $O(\log^2 \Delta \log^2 n)$ in expectation, thus, proving Theorem~\ref{thm:MM}.

\subsection{Data Structures}\label{sec:dsMM}

We maintain the following data structures on each edge $e$ in graph $G$.

\begin{highlighttechnical}
	\begin{itemize}[leftmargin=15pt]
		\item $m(e)$: A binary variable that is 1 if edge $e \in \lfmm{G, \pi}$ and 0 otherwise.
		\item $k(e)$: The rank of $e$'s eliminator, i.e., $k(e) = \pi(\eliminator{G}{\pi}{e})$. Note that $m(e)=1$ iff $k(e) = \pi(e)$.
	\end{itemize}
\end{highlighttechnical}

Furthermore, for any vertex $v$, we maintain the following data structures.

\begin{highlighttechnical}
	\begin{itemize}[leftmargin=15pt]
		\item $k(v)$: If an edge $e \in \lfmm{G, \pi}$ is connected to $v$, then $k(v) = \pi(e)$; otherwise, $k(v) = \infty$.
		\item $N(v)$: The set of edges connected to vertex $v$. The set $N(v)$ is stored as a self-balancing binary search tree and each edge $e$ in it is indexed by $k(e)$.
	\end{itemize}
\end{highlighttechnical}

Similar to MIS, in the pre-processing step, we can simply construct the LFMM of the original graph $G_0=(V, E_0)$ and fill in the data structures above in $O((|V|+|E_0|)\log n)$ time.

\subsection{The Algorithm}\label{sec:algMM}

The following observation is analogous to Observation~\ref{obs:inAiffalowerchanges} for MIS and motivates the same iterative approach in determining the changes in MM.

\newcommand{\obstwoparts}[0]{For any edge $e \in \affectedset$, the following properties hold:
\begin{enumerate}[itemsep=0pt,topsep=7pt]
	\item $k_{t-1}(e) \geq \pi(f)$ and $k_{t}(e) \geq \pi(f)$.
	\item if $e \neq f$, then $e$ has a neighbor $e'$ such that $\pi(e') < \pi(e)$ and $e' \in \flippedset$.
\end{enumerate}}
\begin{observation} \label{obs:twoparts}
\obstwoparts{}
\end{observation}

Algorithm~\ref{alg:edgeupdate} formalizes how our data structures can be updated after each edge insertion/deletion. The subroutines not formalized in the algorithm will be formalized subsequently.

\newcommand{\updatedatastructues}[1]{\textsc{UpdateDataStructures}{\ensuremath{(#1)}}}

\begin{tboxalg2e}{Maintaining the data structures after insertion or deletion of an edge $f=(a ,b)$.}
\label{alg:edgeupdate}
	\begin{algorithm}[H]
	\DontPrintSemicolon
	\SetAlgoSkip{bigskip}
	\SetAlgoInsideSkip{}

	$\S \gets \{f\}$\;
	\While{\S{} is not empty}{
		Let $ e=(u,v) \gets \argmin_{e' \in \S} \pi(e')$ be the minimum rank edge in $\S$.\;
					\updatedatastructues{e} \tcp*{Updates $k(e)$, $m(e)$, $k(v)$, $k(u)$, \A, and \flippedset.}
		\If{$e\in \flippedset$\label{line:ifinFMM}}{
			$\mathcal{H}_e \gets \{ e' \in N(v) \cup N(u) \mid  k_{t-1}(e') \geq \pi(f)\}$\; \tcp{It can be found in time $\O{\log \Delta \cdot |\mathcal{H}_e|}$ since $N(v)$, and $N(u)$ are indexed by $k(.)$.} 
							\For{any edge $e' \in \mathcal{H}_e$ with $\pi(e') > \pi(e)$}{
						insert $e'$ to \S{}.}			
		}
		Remove $e$ from $\S{}$. \label{line:removeS} \; }
	\updateadjlists{} \tcp*{Updates adjacency lists where necessary.}\label{line:upadteNinMM}
	\end{algorithm}
\end{tboxalg2e}
We use {\em iteration} to refer to iterations of the while loop in Algorithm~\ref{alg:edgeupdate}. The following invariants will hold throughout the algorithm.

\begin{invariant}\label{inv:MMS1}
	Consider the start of any iteration and let $e$ be the lowest-rank vertex in $\S$. It holds true that $k(e') = k_t(e')$ and $m(e') = m_t(e')$ for any edge $e'$ with $\pi(e') < \pi(e)$, i.e., $k(e')$ and $m(e')$ already hold the correct values. Moreover, $k(e') = k_{t-1}(e')$ and $m(e') = m_{t-1}(e')$ for every other edge $e'$ with $\pi(e') \geq \pi(e)$.
\end{invariant}

\begin{invariant}\label{inv:MMS2}
	Consider any vertex $v$ in an arbitrary iteration of the algorithm, and let $M_v = \{  e\in E \mid  m(e)=1\}$. 
	Throughout the algorithm, it holds that if $M_v \neq \emptyset$, then $k(v) = \min_{e\in M_v} \pi(e)$, and otherwise   $k(v) = \infty$.
\end{invariant}

We continue by formalizing all subroutines used in Algorithm~\ref{alg:edgeupdate}.

\smparagraph{Subroutine {\normalfont \updatedatastructues{e}}.} Let $u$ and $v$ denote the two end-points of edge $e$. This function updates  $k(e)$, $m(e)$, $k(v)$, and $k(u)$ which also determines the membership of $e$ to sets $\A$ and $\flippedset$. Let $x = \min(k(v), k(u))$. We show that  $e$ joins the matching iff $x\geq \pi(e)$ which results in $m(e) \gets 1$, $k(e) \gets \pi(e)$, $k(v) \gets \pi(e)$, and $k(u) \gets \pi(e)$. Otherwise, we have  $m(e) \gets 0$ and $k(e) \gets x$. Note that if $e$ was previously in the matching and is flipped now, we need to update $k(v)$ and $k(u)$ if they are equal to $\pi(e)$. We show that if  $e$ is removed from the matching and $k(v) = \pi(e)$ then we should set $k(v) \gets \infty$ and the same for vertex $u$.

\smparagraph{Subroutine {\normalfont \updateadjlists{}}.} We first update $N(a)$ and $N(b)$. We remove $f$ from both these sets if $f$ is deleted and add it otherwise.	
	Also, for any affected edge $e = (u, v)$ we need to update its index in sets $N(v)$ and $N(u)$. We do so by a single iteration over set $\A$. Due to the fact that adjacency lists are BSTs with size $O(\Delta)$, this takes  $\O{|\A| \log \Delta}$ time. 

\subsection{Correctness \& (Parametrized) Running Time}

The correctness of Algorithm~\ref{alg:edgeupdate} follows from basic arguments and the greedy structure of LFMM and hence we defer it to Section~\ref{sec:proofsMM}. Here, we discuss why the running time of the algorithm is $O\big(|\flippedset| \min\{\Delta, \frac{\log n}{\pi(f)}\}\log \Delta \big)$ as claimed in Lemma~\ref{lem:updatealgmatching}. The complete proof of both the correctness and running time of the algorithm is presented in Section~\ref{sec:proofsMM}.

Using a similar argument used for MIS, we can use Proposition~\ref{prop:sparseMM} to prove (see Section~\ref{sec:proofsMM}):

\newcommand{\claimsizeofhe}[0]{At any iteration $i$, with probability at least $1-n^{-(c+1)}$, set $\mathcal{H}_e$ has size $O(\min\{\Delta, \frac{\log n}{\pi(f)}\})$ and can be constructed in time $O(|\mathcal{H}_e| \log \Delta)$.}

\begin{claim}\label{claim:sizeofhe}
	\claimsizeofhe{} 
\end{claim}

Let us first analyze the running time before the last line where we update adjacency lists. Observe that any edge $e'$ that is added to set $\S$ belongs to $\mathcal{H}_e$ of an edge $e \in \flippedset$. Therefore, at most $O\big(|\flippedset| \min\{\Delta, \frac{\log n}{\pi(f)}\}\big)$ edges are added to $\S$. Note that, if an edge $e' \in \S$ is not in set \flippedset{}, we only spend $O(1)$ time for it in subroutine \updatedatastructues{e'}. Thus, the total time spent on all edges not in $\flippedset$ is indeed $O\big(|\flippedset| \min\{\Delta, \frac{\log n}{\pi(f)}\}\big)$. On the other hand, for each edge $e \in \flippedset$, the most expensive operation is to find set $\mathcal{H}_e$ which Claim~\ref{claim:sizeofhe} shows can be done in $O(|\mathcal{H}_e|\log \Delta)$ time. Therefore, the total running time before \updateadjlists{} can be bounded by $O\big(|\flippedset| \min\{\Delta, \frac{\log n}{\pi(f)}\}\log \Delta \big)$.

Next, in the \updateadjlists{}, we only iterate over all edges in \affectedset{} and update their position in their end-points. This takes $O(|\affectedset|\log \Delta)$ time. Note that by Observation~\ref{obs:twoparts}, any edge $e' \in \affectedset$ is adjacent to an edge $e \in \flippedset$ and $k_{t-1}(e') \geq \pi(f)$. This means that $e' \in \mathcal{H}_e$ and by Claim~\ref{claim:sizeofhe}:

\newcommand{\clAisFtimesH}[0]{$|\affectedset| \leq O\big(|\flippedset| \min\{\Delta, \frac{\log n}{\pi(f)} \} \big)$.}

\begin{observation}\label{cl:AisFtimesH}
	W.h.p., \clAisFtimesH{}
\end{observation}

Therefore, the overall running time is indeed $O\big(|\flippedset| \min\{\Delta, \frac{\log n}{\pi(f)}\}\log \Delta \big)$ as claimed in Lemma~\ref{lem:updatealgmatching}.

\subsection{Putting Everything Together: Proof of Theorem~\ref{thm:MM}}\label{sec:analysisMM}

Before proving Theorem~\ref{thm:MM} we need the following high probability bound of $O(\log n)$ on $|\flippedset|$ which we prove in Section~\ref{sec:proofsMM}.

\newcommand{\clboundFMM}[0]{Let $G$ and $G'$ be two graphs that differ in only one edge and let $\pi$ be a random ranking on their edges. Then, w.h.p., there are at most $O(\log n)$ edges that have different MM-statuses in \lfmm{G, \pi} and \lfmm{G', \pi}.}

\begin{claim}\label{cl:boundFMM}
	\clboundFMM{}
\end{claim}

Now, we are ready to prove Theorem~\ref{thm:MM}.

\restate{Theorem~\ref{thm:MM}}{\thmMM{}}

\begin{proof}
	We use Algorithm~\ref{alg:edgeupdate}. Combination of Lemma~\ref{lem:updatealg}, and the fact that $|\flippedset| \leq O(\log n)$ w.h.p. due to Claim~\ref{cl:boundFMM}, bounds the  update-time of this algorithm, w.h.p., by
	\begin{flalign*}
	O(\log \Delta \log n) \min\left\{\Delta, \frac{\log n}{\pi(f)}\right\} = O(\log \Delta \log^2 n) \min\left\{\Delta, \frac{1}{\pi(f)}\right\}.
	\end{flalign*}
	Since $\pi(f)$ is chosen from $[0, 1]$ uniformly at random, $\E\big[\min\big\{\Delta, \frac{1}{\pi(f)}\big\}\big] = O(\log \Delta)$. Thus, the total running time is $O(\log^2 \Delta \log^2 n)$ in expectation, as required by the theorem.
	
	For the adjustment-complexity, similar to MIS, it is shown in \cite[Theorem~1]{Censor-HillelHK16} that LFMIS over random rankings requires $O(1)$ expected adjustments under {\em vertex} updates. On the line-graph, this implies that if an edge is added or removed, the number of changes to LFMM over a random ranking is $O(1)$ in expectation; concluding the proof.
\end{proof}

\subsection{Deferred Proofs}\label{sec:proofsMM}

\restate{Observation~\ref{obs:twoparts}}{\obstwoparts{}}

\begin{proof}[Proof of  part 1]
Let $U$ denote the set of edges $e$ in $E$ with $\pi(e) < \pi(f)$. 
Consider the subgraph only containing these edges. Since the matching is constructed greedily on the lower rank edges first, the  
set of matching edges in  $U$ does not change after the update.
Let $M_U$ denote the matching edges in $U$. Note that any edge $e$ with $k_{t-1}(e) < \pi(f)$ is incident to an edge $e'$ in $M_U$. Since $e$ and $e'$  are still incident after the update and that updating $f$ does not change $k(e')$ we have $k_t(e) = k_{t-1}(e)$. This means that for each edge $e$ with $k_{t-1}(e) < \pi(f)$, we have $k_t(e) = k_{t-1}(e)$ and thus $e$ cannot be in $\affectedset$ by definition.  By a similar argument, for any edge $e$ with $k_t(e) < \pi(f)$ we also have $k_{t-1}(e)  = k_t(e)$ and thus $e \not\in \affectedset{}$.
\end{proof}

\begin{proof}[Proof of  part 2]

The fact that $e\in \A$ means that eliminator of edge $e$ changes after the update. Let $e'$ be its eliminator before the update. By definition of the eliminator, for $e\neq f$ we have $e\in \A$ iff the matching status of at least an edge incident to $e$ with rank at most $\pi(e')$ changes. This means that if $e$ is not incident to any edge in $\flippedset$, then $e\notin \A$.
\end{proof}

\begin{claim}\label{claim:updatedata}
Let $e = (u, v)$ be the lowest-rank edge in $\S$ at the start of an arbitrary iteration. Subroutine \updatedatastructues{e} correctly updates $k(e)$ and $m(e)$ in constant time  assuming that Invariants~\ref{inv:MMS1} and \ref{inv:MMS2} hold by this iteration.
\end{claim}

\begin{proof}
By definition, we know that eliminator of edge $e$ is its lowest-rank edge in $N(v) \cup N(u)$ that is in the matching after the update. By Invariants~\ref{inv:MMS2} $ \min(k(u), k(v))$ is the rank of an edge who has the lowest-rank amongst the edges $e'$ in $N(v) \cup N(u)$ with $m(e')=1$. Moreover, by  Invariants~\ref{inv:MMS1}, we know that for any edge $e'$ with $\pi(e') < \pi(e)$, we have $m(e') = m_{t}(e')$. This means that $ \min(k(u), k(v)) < \pi(e)$ iff there is at least one edge adjacent to $e$ that is in the matching after the update. In the subroutine \updatedatastructues{e}, we use this condition to determine $m(e)$. Further in the subroutine if $m(e)=1$ we set $k(e) = \pi(e)$ and otherwise set it to $\min(k(u), k(v))$ which is correct by definition of eliminator. To sum up,  subroutine \updatedatastructues{e} correctly updates $k(e)$ and $m(e)$ for edge $e$ the lowest-rank edge in $\S$.
\end{proof}

\begin{observation}\label{claim:increasingedges}
	Let $e$ and $e'$  respectively denote two edges removed from $\S$ in two consecutive iteration of the algorithm in Line~\ref{line:removeS}. We have $\pi(e) < \pi(e')$. 
\end{observation}

\begin{proof}
Let $\S_{i}$ and $\S_{i+1}$ respectively denote set $\S$ at the beginning of iteration $i$ and set $\S$ at the beginning of iteration  $i+1$ and let $e_i$ and $e_{i+1}$ be the lowest-rank edges in these sets. 
We know that $e_{i+1}$ is either inserted to $\S$ in iteration $i$ or that it is in set $\S_{i}$. Observe that any edge added to $\S$ in the $i$-th iteration has rank lower than $\pi(e_i)$ and that $e_i$ is the lowest-rank vertex in $\S_{i}$. As a result $\pi(e_i) < \pi(e_{i+1})$.	
\end{proof}

\begin{claim} \label{claim:inv1inv2}
Let $e$ be the lowest-rank edge in $\S$ in an arbitrary iteration $i$ of the algorithm. Assuming that Invariants~\ref{inv:MMS1}, and \ref{inv:MMS2} hold at the start of iteration $i$ we have:
\begin{enumerate}
	\item If $\S = \emptyset$ at the end of iteration $i$, then for any edge $e$ we have $k(e)= k_{t}(e)$, and $m(e)=m_{t}(e)$ and for any vertex $v$ we have $k(v) = k_t(v)$.
	\item If  $\S \neq \emptyset$ at the end of iteration $i$, then Invariants~\ref{inv:MMS1} and \ref{inv:MMS2} hold at the start of iteration $i+1$ as well.   
\end{enumerate}

\end{claim}

\begin{proof}
	By Claim~\ref{claim:updatedata}, we know that by the end of iteration $i$, for any edge $e'$ with $\pi(e') \leq \pi(e)$ we have $m(e') = m_t(e')$, and $k(e') = k_t(e')$. Let $g$ be the lowest-rank edge in $\A$ whose  $k(g)$ or  $m(g)$ are not updated at the end of iteration $i$. We show that if such an edge exists it is in set $\S$. Note that by Observation~\ref{obs:twoparts}, edge $g$ has at least one incident edge $g'$ where $\pi(g')< \pi(g)$ and $g'\in \flippedset$. Since $g$ is the lowest-rank edge whose  $k(g)$ or  $m(g)$ are not updated, we get that $k(g')$ or  $m(g')$ are updated. This means that there was an iteration  $j < i$ of the algorithm in which $g'$ was the lowest-rank edge in $\S$ since $k_t(g') = k_{t-1}(g')$ and that in each iteration we only update $k_t()$ for the lowest-rank edge in $\S$. Since $g'$ is in $\flippedset$ in iteration $j$, the algorithm adds all the edges in $\mathcal{H}_e$ to set $\S$ in that iteration. By definition of $\mathcal{H}_e$, this set includes edge $g$. Also, note that by Observation~\ref{claim:increasingedges}, the rank of vertices removed from set $\S$ is increasing; thus $g$ is still in set $\S$ in iteration $i$. This means that if set $\S$ is empty then for all edges $g$, we have $m(g) = m_t(g)$ and $k(g)=k_t(g)$ and if it is nonempty Invariants~\ref{inv:MMS1} holds in the next iteration. 
	
	To complete the proof it suffices to show that if $\S$ is empty at the end of iteration $i$, for any vertex $v$ we have $k(v) = k_t(v)$ and that otherwise Invariants~\ref{inv:MMS2} still holds at iteration $i+1$. Note that in the $i$-th iteration we do not change $m(e')$ if $e'\neq e$. Thus, given that Invariants~\ref{inv:MMS2} holds at the beginning of iteration $i$, for any vertex $v$ that is not incident to $e$ we have $k = k_t(v)$ at the end of the iteration as well. Now consider vertex $u$ that is incident to $e$. If $e$ is not flipped  or if $k(u)< \pi(e)$ the algorithm does not change $k(u)$ which is correct by definition of $k(u)$. Therefore, we only need to consider the case that  $e$ is flipped and $k(u) \geq \pi(e)$. In this case, if $m_t(e)=1$, the it is be the lowest-rank edge adjacent to $u$ with $m(.)=1$. Algorithm correctly detects this and sets $k(u)= \pi(e)$ in this scenario. Further, if $m_t(e)=0$ (which means $m_{t-1}(e)=1$), then there is no other edge adjacent to $u$ with $m(.)=1$ in which case, as well, the algorithm correctly sets $k(u)=\infty$. We achieved this from the fact that each vertex has at most one edge with $m_{t-1}(.) = 1$ and by Invariants~\ref{inv:MMS1}  any edge $u_1$ with a higher rank than $u$ has $k(u_1) = k_{t-1}(u_1)$. To sum up, Invariant~\ref{inv:MMS2} still holds at the end of iteration $i$ and the proof of the claim is completed.
\end{proof}

%

\restate{Claim~\ref{claim:sizeofhe}}{\claimsizeofhe}

\begin{proof}
\smparagraph{Size of $\mathcal{H}_e$.} Observe that if $\mathcal{H}_e$ is defined, then as assured by the condition in Line~\ref{line:ifinFMM} of Algorithm~\ref{alg:edgeupdate}, $e \in \flippedset \subseteq \affectedset$ thus by Observation~\ref{obs:inAiffalowerchanges}, $k_{t-1}(e) \geq \pi(f)$. Furthermore, by definition, every edge $e' \in \mathcal{H}_e$ has $k_{t-1}(e') \geq \pi(f)$. This means that if we take LFMM of $G_{t-1}$ induced on edges with rank in $[0, \pi(f))$ and remove them and their neighbors from the graph, $e$ and all of its neighbors in $\mathcal{H}_e$ will survive. Recall that the adversary is oblivious and the graph $G_{t-1}$ and random permutation $\pi$ are chosen independently. Therefore, applying Proposition~\ref{prop:sparseMM} on graph $G_{t-1}$ with $p = \pi(f)$ bounds $|\mathcal{H}_e|$ by $O(\pi(f)^{-1} \log n)$ w.h.p. Moreover, clearly $|\mathcal{H}_e| \leq 2\Delta-2$ since all edges in it are incident to $e$, concluding the bound on the size of $\mathcal{H}_e$.	

\smparagraph{Construction of $\mathcal{H}_e$.} Note that we do not change the adjacency lists $N(.)$ stored on the vertices until the very last line of Algorithm~\ref{alg:edgeupdate}. Therefore, for any vertex $v$, we have $N(v) = N_{t-1}(v)$ before this line. This means that throughout the algorithm, for any edge $e=(u, v)$ we can iterate over edges in $N(u)$ and $N(v)$ and find all edges $e'$ with $k_{t-1}(e') \geq \pi(a)$; all these edges will belong to $\mathcal{H}_e$. Thus the total time required is $O(|\mathcal{H}_e| \log \Delta)$. Note that this is possible since $N(v)$ and $N(u)$ are BSTs indexed by $k_{t-1}(.)$ of the elements in them but comes at the cost of an extra $O(\log \Delta)$ factor as these BSTs can have size up to $\Delta$.
\end{proof}

Combining all these claims, we can prove Lemma~\ref{lem:updatealgmatching}.

\restate{Lemma~\ref{lem:updatealgmatching}}{There is an algorithm to update \lfmm{G, \pi} and the data structures required for it after insertion or deletion of any edge $f=(a, b)$ in $O\left(|\flippedset| \min\{\Delta, \frac{\log n}{\pi(f)} \} \log\Delta  \right)$ time, w.h.p.}


\begin{proof}	
\smparagraph{Correctness.}
We first show that when set $S$ becomes empty, for any edge $e$ we have $k(e)= k_{t}(e)$, and $m(e)=m_{t}(e)$ and for any vertex $v$ we have $k(v) = k_t(v)$. To do so, we will use proof by induction and show that Invariants~\ref{inv:MMS1} and \ref{inv:MMS2} hold throughout the algorithm. This proves our claim since by Observation~\ref{claim:increasingedges} if both invariants hold in the last iteration of the algorithm, then when set $S$ becomes empty the data structures $k(.)$ and $m(.)$ are updated for all edges and vertices. By Observation~\ref{obs:twoparts}, for any edge $e'$ with $\pi(e') < \pi(f)$ we have $m(e') = m_t(e')$ and $k(e')= k_t(e')$ which means that Invariant~\ref{inv:MMS1} holds in the first iteration. Further, Invariant~\ref{inv:MMS2} holds since $m(.)$ of none of the edges has changed yet. This gives us the base case of the induction. Moreover, the induction step is a direct result of Claim~\ref{claim:inv1-inv2} which states that if both invariants hold in an arbitrary iteration they hold in the next iteration given that $\S$ is nonempty. 

To complete the prove of correctness, we need to show that when the algorithm terminates, for any vertex $v$, we have $N(v)= N_t(v)$.  
Subroutine {\normalfont \updateadjlists{}} first modifies the adjacency lists of vertices $a$ and $b$ by adding $e$ to them if $e$ is to be added or deleting it otherwise. Note that as we showed, when the algorithm runs this subrutine, for any edge $e$ we already have $k(e)= k_{t-1}(e)$, thus set $\A$ is also updated. Therefore, Algorithm~\ref{alg:edgeupdate}, correctly updates the adjacency lists by iterating over edges in $\A$ and updating their index in the adjacency lists of their end-points.  

\smparagraph{Running Time.} First, note that by Claim~\ref{claim:sizeofhe}, with probability at least $1-n^{-c}$, the size of $\mathcal{H}_e$ for any edge $e$ is $O(\min\{\Delta, \frac{\log n}{\pi(f)}\})$ and constructing that takes time $O(\log \Delta \cdot \min\{\frac{\log n}{\pi(f)}, \Delta \})$. Moreover, by Observation~\ref{claim:increasingedges}, we know that each edge is the lowest-rank edge in set $\S$ in at most one iteration. Putting these facts together gives us that the number of iterations of the algorithm is $O(|\flippedset |\cdot \min\{\frac{\log n}{\pi(f)}, \Delta \})$ with probability at least $1-n^{-c}$. We also know by Claim~\ref{claim:updatedata} that subroutine \updatedatastructues{e} takes $O(1)$ time. Therefore,  the total running time of the algorithm until set $\S$ becomes empty is $O(|\flippedset | \cdot \log \Delta \cdot \min\{\frac{\log n}{\pi(f)}, \Delta \})$ with probability at least $1-n^{-c}$. Further, subroutine {\normalfont \updateadjlists{}} takes $O(|\A|\log \Delta)$ time which by Claim~\ref{cl:AisFtimesH} is bounded by 
$O(\log \Delta \cdot \min\{\frac{\log n}{\pi(f)}, \Delta \})$ with probability at least $1-n^{-c}$. 
\end{proof}

\restate{Claim~\ref{cl:boundFMM}}{\clboundFMM{}}

\begin{proof}
	Assume without loss of generality that $G'$ is obtained by removing an edge $e$ from $G$ and let $\flippedset$ be the set of edges with different MM-statuses in \lfmm{G, \pi} and \lfmm{G', \pi}. We first show that: (1) Each edge $e \in \flippedset$ with $e \not= f$, has a lower-rank neighboring edge in $\flippedset$. (2) The edges in $\flippedset$ form either a single path or a single cycle.
	
	Proof of (1) directly follows from Observation~\ref{obs:twoparts} part 2. For (2), observe that each edge in $\flippedset$ is in at least one of the two matchings \lfmm{G, \pi} and \lfmm{G', \pi}. Therefore, each vertex has at most two incident edges in $\flippedset$; meaning that each connected component in \flippedset{} is indeed either a cycle or a path. To see why there cannot be more than one such connected component,  observe that in this case, at least one connected component does not include $f$. Let $g$ be the minimum-rank edge in this component. For $g$, (1) cannot hold which is a contradiction. 
	
	Now, we show that $|\flippedset| = O(\log n)$. To do this, we provide a reduction to the parallel round complexity of LFMM over random orders.
	
	LFMM can be parallelized, just like LFMIS as described in Section~\ref{sec:unlikely}, in the following way: In each round, all edges that hold the locally minimum rank among their neighbors join MM, then we remove them and their neighboring edges. It is known from \cite[Corollary~B.1]{DBLP:conf/soda/FischerN18} that if ranking $\pi$ over the edges is chosen randomly, then it takes $O(\log n)$ rounds until we find a maximal matching, with probability $1-n^{-c}$ for any constant $c > 1$. 
	
	We prove that the parallel round-complexity of random-order LFMM is at least $\Omega(|\flippedset|)$, implying that w.h.p. $|\flippedset| = O(\log n)$ as desired. To do this, observe that by properties (1) and (2) above, there should be a monotone path $P=(e_1, \ldots, e_k)$ in \flippedset{} where $\pi(e_i) < \pi(e_{i+1})$ for any $i \in [k-1]$ and where  $k = \Omega(|\flippedset|)$. (Just take the longest path in $\flippedset \setminus \{ f \}$, by (1) it has size $\Omega(|\flippedset|)$ and by (2) it is monotone.) Furthermore, since each edge in $P$ is in $\flippedset$ and the edges in $\flippedset$ belong to exactly one of \lfmm{G, \pi} and \lfmm{G', \pi}, the edges in $P$ have to alternate between the two matchings. Suppose w.l.o.g. that the odd ones belong to \lfmm{G, \pi}. Now, take edge $w_{2i+1}$ for any $i$. We show that it takes at least $i$ parallel rounds until this edge joins \lfmm{G, \pi}. For $w_{2i+1}$ to join the matching, its lower rank neighbor $w_{2i}$ should be removed so that $w_{2i+1}$ becomes the local minimum edge. This does not happen until $w_{2i-1}$ joins the matching since $w_{2i-1}$ and $w_{2i+1}$ are the only incident edges to $w_{2i}$ that are in \lfmm{G, \pi}. Now, a simple induction implies that it takes at least $i$ rounds until $w_{2i+1}$ joins the matching, and thus the parallel round complexity is at least $\Omega(k) = \Omega(|\flippedset|)$, which as described, implies $|\flippedset| = O(\log n)$ w.h.p.
\end{proof}

\bibliographystyle{plain}
\bibliography{refs}

\begin{thebibliography}{10}

\bibitem{DBLP:conf/icml/AhnCGMW15}
Kook~Jin Ahn, Graham Cormode, Sudipto Guha, Andrew McGregor, and Anthony Wirth.
\newblock {Correlation Clustering in Data Streams}.
\newblock In {\em Proceedings of the 32nd International Conference on Machine
  Learning, {ICML}}, pages 2237--2246, 2015.

\bibitem{correlationclustering}
Nir Ailon, Moses Charikar, and Alantha Newman.
\newblock {Aggregating inconsistent information: Ranking and clustering}.
\newblock {\em J. {ACM}}, 55(5):23:1--23:27, 2008.

\bibitem{DBLP:journals/jal/AlonBI86}
Noga Alon, L{\'{a}}szl{\'{o}} Babai, and Alon Itai.
\newblock {A Fast and Simple Randomized Parallel Algorithm for the Maximal
  Independent Set Problem}.
\newblock {\em J. Algorithms}, 7(4):567--583, 1986.

\bibitem{DBLP:conf/icalp/ArarCCSW18}
Moab Arar, Shiri Chechik, Sarel Cohen, Cliff Stein, and David Wajc.
\newblock {Dynamic Matching: Reducing Integral Algorithms to
  Approximately-Maximal Fractional Algorithms}.
\newblock In {\em 45th International Colloquium on Automata, Languages, and
  Programming, {ICALP} 2018, July 9-13, 2018, Prague, Czech Republic}, pages
  7:1--7:16, 2018.

\bibitem{assadistoc}
Sepehr Assadi, Krzysztof Onak, Baruch Schieber, and Shay Solomon.
\newblock {Fully Dynamic Maximal Independent Set with Sublinear Update Time}.
\newblock In {\em Proceedings of the 50th Annual {ACM} {SIGACT} Symposium on
  Theory of Computing, {STOC} 2018, Los Angeles, CA, USA, June 25-29, 2018},
  pages 815--826, 2018.

\bibitem{assadisoda}
Sepehr Assadi, Krzysztof Onak, Baruch Schieber, and Shay Solomon.
\newblock {Fully Dynamic Maximal Independent Set with Sublinear in $n$ Update
  Time}.
\newblock In {\em Proceedings of the Thirtieth Annual {ACM-SIAM} Symposium on
  Discrete Algorithms, {SODA} 2019, San Diego, California, USA, January 6-9,
  2019}, pages 1919--1936, 2019.

\bibitem{DBLP:journals/siamcomp/BaswanaGS18}
Surender Baswana, Manoj Gupta, and Sandeep Sen.
\newblock {Fully Dynamic Maximal Matching in $O(\log n)$ Update Time (Corrected
  Version)}.
\newblock {\em {SIAM} J. Comput.}, 47(3):617--650, 2018.

\bibitem{DBLP:journals/corr/abs-1901-03744}
Soheil Behnezhad, MohammadTaghi Hajiaghayi, and David~G. Harris.
\newblock {Exponentially Faster Massively Parallel Maximal Matching}.
\newblock {\em CoRR}, abs/1901.03744, 2019.

\bibitem{DBLP:conf/soda/BernsteinFH19}
Aaron Bernstein, Sebastian Forster, and Monika Henzinger.
\newblock {A Deamortization Approach for Dynamic Spanner and Dynamic Maximal
  Matching}.
\newblock In {\em Proceedings of the Thirtieth Annual {ACM-SIAM} Symposium on
  Discrete Algorithms, {SODA} 2019, San Diego, California, USA, January 6-9,
  2019}, pages 1899--1918, 2019.

\bibitem{DBLP:conf/ipco/BhattacharyaCH17}
Sayan Bhattacharya, Deeparnab Chakrabarty, and Monika Henzinger.
\newblock {Deterministic Fully Dynamic Approximate Vertex Cover and Fractional
  Matching in {$O(1)$} Amortized Update Time}.
\newblock In {\em Integer Programming and Combinatorial Optimization - 19th
  International Conference, {IPCO} 2017, Waterloo, ON, Canada, June 26-28,
  2017, Proceedings}, pages 86--98, 2017.

\bibitem{DBLP:journals/siamcomp/BhattacharyaHI18}
Sayan Bhattacharya, Monika Henzinger, and Giuseppe~F. Italiano.
\newblock {Deterministic Fully Dynamic Data Structures for Vertex Cover and
  Matching}.
\newblock {\em {SIAM} J. Comput.}, 47(3):859--887, 2018.

\bibitem{DBLP:conf/stoc/BhattacharyaHN16}
Sayan Bhattacharya, Monika Henzinger, and Danupon Nanongkai.
\newblock {New deterministic approximation algorithms for fully dynamic
  matching}.
\newblock In {\em Proceedings of the 48th Annual {ACM} {SIGACT} Symposium on
  Theory of Computing, {STOC} 2016, Cambridge, MA, USA, June 18-21, 2016},
  pages 398--411, 2016.

\bibitem{DBLP:conf/soda/BhattacharyaHN17}
Sayan Bhattacharya, Monika Henzinger, and Danupon Nanongkai.
\newblock {Fully Dynamic Approximate Maximum Matching and Minimum Vertex Cover
  $O(\log^3 n)$ Worst Case Update Time}.
\newblock In {\em Proceedings of the Twenty-Eighth Annual {ACM-SIAM} Symposium
  on Discrete Algorithms, {SODA} 2017, Barcelona, Spain, Hotel Porta Fira,
  January 16-19}, pages 470--489, 2017.

\bibitem{DBLP:conf/spaa/BlellochFS12}
Guy~E. Blelloch, Jeremy~T. Fineman, and Julian Shun.
\newblock {Greedy sequential maximal independent set and matching are parallel
  on average}.
\newblock In {\em 24th {ACM} Symposium on Parallelism in Algorithms and
  Architectures, {SPAA} '12, Pittsburgh, PA, USA, June 25-27, 2012}, pages
  308--317, 2012.

\bibitem{Censor-HillelHK16}
Keren Censor{-}Hillel, Elad Haramaty, and Zohar~S. Karnin.
\newblock {Optimal Dynamic Distributed {MIS}}.
\newblock In {\em Proceedings of the 2016 {ACM} Symposium on Principles of
  Distributed Computing, {PODC} 2016, Chicago, IL, USA, July 25-28, 2016},
  pages 217--226, 2016.

\bibitem{DBLP:conf/icalp/CharikarS18}
Moses Charikar and Shay Solomon.
\newblock {Fully Dynamic Almost-Maximal Matching: Breaking the Polynomial
  Worst-Case Time Barrier}.
\newblock In {\em 45th International Colloquium on Automata, Languages, and
  Programming, {ICALP} 2018, July 9-13, 2018, Prague, Czech Republic}, pages
  33:1--33:14, 2018.

\bibitem{independentwork}
Shiri Chechik and Tianyi Zhang.
\newblock {Fully Dynamic Maximal Independent Set in Expected Poly-Log Update
  Time}.
\newblock In {\em 49th Annual {IEEE} Symposium on Foundations of Computer
  Science, {FOCS} 2019, to appear}, 2019.

\bibitem{DBLP:conf/podc/DaumGKN12}
Sebastian Daum, Seth Gilbert, Fabian Kuhn, and Calvin~C. Newport.
\newblock {Leader election in shared spectrum radio networks}.
\newblock In {\em {ACM} Symposium on Principles of Distributed Computing,
  {PODC} '12, Funchal, Madeira, Portugal, July 16-18, 2012}, pages 215--224,
  2012.

\bibitem{DBLP:journals/corr/abs-1804-08908}
Yuhao Du and Hengjie Zhang.
\newblock {Improved Algorithms for Fully Dynamic Maximal Independent Set}.
\newblock {\em CoRR}, abs/1804.08908, 2018.

\bibitem{DBLP:conf/soda/FischerN18}
Manuela Fischer and Andreas Noever.
\newblock {Tight Analysis of Parallel Randomized Greedy {MIS}}.
\newblock In {\em Proceedings of the Twenty-Ninth Annual {ACM-SIAM} Symposium
  on Discrete Algorithms, {SODA} 2018, New Orleans, LA, USA, January 7-10,
  2018}, pages 2152--2160, 2018.

\bibitem{DBLP:conf/podc/GhaffariGKMR18}
Mohsen Ghaffari, Themis Gouleakis, Christian Konrad, Slobodan Mitrovic, and
  Ronitt Rubinfeld.
\newblock {Improved Massively Parallel Computation Algorithms for MIS,
  Matching, and Vertex Cover}.
\newblock In {\em Proceedings of the 2018 {ACM} Symposium on Principles of
  Distributed Computing, {PODC} 2018, Egham, United Kingdom, July 23-27, 2018},
  pages 129--138, 2018.

\bibitem{DBLP:conf/stoc/GuptaK0P17}
Anupam Gupta, Ravishankar Krishnaswamy, Amit Kumar, and Debmalya Panigrahi.
\newblock {Online and dynamic algorithms for set cover}.
\newblock In {\em Proceedings of the 49th Annual {ACM} {SIGACT} Symposium on
  Theory of Computing, {STOC} 2017, Montreal, QC, Canada, June 19-23, 2017},
  pages 537--550, 2017.

\bibitem{DBLP:journals/corr/abs-1804-01823}
Manoj Gupta and Shahbaz Khan.
\newblock {Simple dynamic algorithms for Maximal Independent Set and other
  problems}.
\newblock {\em CoRR}, abs/1804.01823, 2018.

\bibitem{DBLP:conf/focs/GuptaP13}
Manoj Gupta and Richard Peng.
\newblock {Fully Dynamic {$(1+\epsilon)$}-Approximate Matchings}.
\newblock In {\em 54th Annual {IEEE} Symposium on Foundations of Computer
  Science, {FOCS} 2013, 26-29 October, 2013, Berkeley, CA, {USA}}, pages
  548--557, 2013.

\bibitem{DBLP:conf/focs/Linial87}
Nathan Linial.
\newblock {Distributive Graph Algorithms-Global Solutions from Local Data}.
\newblock In {\em 28th Annual Symposium on Foundations of Computer Science, Los
  Angeles, California, USA, 27-29 October 1987}, pages 331--335, 1987.

\bibitem{DBLP:conf/stoc/Luby85}
Michael Luby.
\newblock {A Simple Parallel Algorithm for the Maximal Independent Set
  Problem}.
\newblock In {\em Proceedings of the 17th Annual {ACM} Symposium on Theory of
  Computing, May 6-8, 1985, Providence, Rhode Island, {USA}}, pages 1--10,
  1985.

\bibitem{DBLP:conf/stoc/NeimanS13}
Ofer Neiman and Shay Solomon.
\newblock {Simple deterministic algorithms for fully dynamic maximal matching}.
\newblock In {\em Symposium on Theory of Computing Conference, STOC'13, Palo
  Alto, CA, USA, June 1-4, 2013}, pages 745--754, 2013.

\bibitem{DBLP:conf/focs/NguyenO08}
Huy~N. Nguyen and Krzysztof Onak.
\newblock {Constant-Time Approximation Algorithms via Local Improvements}.
\newblock In {\em 49th Annual {IEEE} Symposium on Foundations of Computer
  Science, {FOCS} 2008, October 25-28, 2008, Philadelphia, PA, {USA}}, pages
  327--336, 2008.

\bibitem{DBLP:conf/stoc/OnakR10}
Krzysztof Onak and Ronitt Rubinfeld.
\newblock {Maintaining a large matching and a small vertex cover}.
\newblock In {\em Proceedings of the 42nd {ACM} Symposium on Theory of
  Computing, {STOC} 2010, Cambridge, Massachusetts, USA, 5-8 June 2010}, pages
  457--464, 2010.

\bibitem{onakicalp}
Krzysztof Onak, Baruch Schieber, Shay Solomon, and Nicole Wein.
\newblock {Fully Dynamic {MIS} in Uniformly Sparse Graphs}.
\newblock In {\em 45th International Colloquium on Automata, Languages, and
  Programming, {ICALP} 2018, July 9-13, 2018, Prague, Czech Republic}, pages
  92:1--92:14, 2018.

\bibitem{DBLP:conf/focs/Solomon16}
Shay Solomon.
\newblock {Fully Dynamic Maximal Matching in Constant Update Time}.
\newblock In {\em {IEEE} 57th Annual Symposium on Foundations of Computer
  Science, {FOCS} 2016, 9-11 October 2016, Hyatt Regency, New Brunswick, New
  Jersey, {USA}}, pages 325--334, 2016.

\bibitem{DBLP:conf/stoc/YoshidaYI09}
Yuichi Yoshida, Masaki Yamamoto, and Hiro Ito.
\newblock {An improved constant-time approximation algorithm for maximum
  matchings}.
\newblock In {\em Proceedings of the 41st Annual {ACM} Symposium on Theory of
  Computing, {STOC} 2009, Bethesda, MD, USA, May 31 - June 2, 2009}, pages
  225--234, 2009.

\end{thebibliography}

\appendix

\section{Proofs of Propositions~\ref{prop:sparseMIS} and \ref{prop:sparseMM}: Degree Pruning}\label{sec:sparsificationproof}

\restate{Proposition~\ref{prop:sparseMIS}}{\sparseMIS{}}

\begin{proof}
	Let us use \rlfmis{p}{G, \pi} to denote the subset of vertices in \lfmis{G, \pi} with rank in $[0, p]$. Any vertex $v$ with $\eliminator{G}{\pi}{v} \leq p$ is in set $\Gamma(\rlfmis{p}{G, \pi})$. Therefore, the set $V_p$ is precisely equal to set $V \setminus \Gamma(\rlfmis{p}{G, \pi})$. Therefore, we have to show that once we remove all vertices in $\Gamma(\rlfmis{p}{G, \pi})$ from the graph, the maximum remaining degree drops to $O(p^{-1}\log n)$ w.h.p.

	Fix an arbitrary $O(\log n)$ bit real $p \in [0, 1]$ and an arbitrary vertex $v \in V$. Let us use $H=(V_H, E_H)$ to denote the residual graph $G[V \setminus \Gamma(\rlfmis{p}{G, \pi})]$ and use $d_v$ to denote the residual degree of $v$ in $H$. That is, $d_v = 0$ if $v \not\in V_H$ and $d_v = \deg_H(v)$ otherwise. The main part of the proof is to show that for any parameter $\beta \geq 1$, it holds that $\Pr[d_v \geq p^{-1} \beta] \leq e^{-\beta}$. Then setting $\beta = \alpha \ln n$, for some large enough constant $\alpha$ would imply $\Pr[d_v \geq p^{-1} \alpha \log n] \leq n^{-\alpha}$. Combining this with a simple union bound over the $\poly(n)$ many possible pairs of $v$ and $p$ will conclude the proof of the proposition that $d_v = O(p^{-1} \log n)$ for all $v$ and $p$ w.h.p.
	
	We first describe a random process for generating an independent set $I$. Then we prove that distribution of \rlfmis{p}{G, \pi} and $I$ is exactly the same and thus both have the same probabilistic behavior. We then prove that $I$ significantly prunes the vertex degrees in the residual graph.

	The random process is as follows: A permutation $\psi$ over all vertices in $V$ is fixed uniformly at random and we initialize $I$ to be $\emptyset$. We then iterate over the vertices in the order of $\psi$. Once we encounter a vertex $v$ in the process, if $v$ has a neighbor in $I$, we discard $v$ and call it {\em irrelevant}. Otherwise,  $v$ is {\em relevant} and we draw a Bernoulli random variable $x_v$ which is 1 with probability $p$. If $x_v = 0$, we call $v$ {\em unlucky} and discard it and if $x_v = 1$, we add $v$ to $I$ and call it {\em lucky}.
		
	It is easy to verify that the random process above in determining $I$ is precisely equivalent to the following process: First sample each vertex of $V$ into a set $S$ independently with probability $p$, then fix a random permutation $\psi$ over the vertices in $S$ and let $I$ be the independent set \lfmis{G[S], \psi}. On the other hand, recall that \rlfmis{p}{G, \pi} is the LFMIS obtained once we only process the vertices $v$ with rank $\pi(v)$ in $[0, p]$. Since $\pi$ is a random ranking, the probability that for a vertex $v$, $\pi(v) \in [0, p]$ is $p$ and is independent from the rank of the other vertices. Furthermore, once we condition on the set of vertices with rank within $[0, p]$, their internal ordering will be completely at random. Therefore, the two independent sets $I$ and $\rlfmis{p}{G, \pi}$ have the same distribution and thus the same probabilistic behavior.

	Finally, we prove the promised claim that $\Pr[d_v \geq p^{-1} \beta] \leq e^{-\beta}$ for any vertex $v$ and any parameter $\beta > 1$ by considering the pruning effect of $I$, which we showed above is equivalent to that of \rlfmis{p}{G, \pi}. To have $d_v \geq p^{-1}\beta$, $v$ should survive to the residual graph, i.e., $v \not \in \Gamma(I)$. This means that in the original process for constructing $I$, anytime that we encounter a relevant neighbor $u$ of $v$, it should turn out to be unlucky. Furthermore, the irrelevant neighbors of $v$ do not survive to the residual graph. Therefore, to have $d_v \geq p^{-1}\beta$, we should encounter at least $p^{-1}\beta$ relevant neighbors of $v$. The probability that all these neighbors turn out to be unlucky is $(1-p)^{\beta/p} \leq e^{-\beta}$ as desired.
\end{proof}

\restate{Proposition~\ref{prop:sparseMM}}{\sparseMM{}}

\begin{proof}
	This simply follows by applying Proposition~\ref{prop:sparseMIS} to the line-graph of $G$.
	
	Consider the line-graph $L=(V^L, E^L)$ of graph $G$, i.e., $V^L$ is equivalent to $E$. Note that ranking $\pi$ on $E$ is a random ranking on the vertices of $L$. Moreover, for any edge $e \in E$ and its equivalent vertex $v_e \in V_L$, their eliminators \eliminator{G}{\pi}{e} and \eliminator{L}{\pi}{v_e} are also equivalent (due to the well-known equivalence of MM and MIS on the line-graph). This means that $e \in E_p$, iff $v_e \in V^L_p$ where $V^L_p = \{ v \in V^L \mid \pi(\eliminator{L}{\pi}{v}) > p \}$. By Proposition~\ref{prop:sparseMIS} we know the maximum degree in $L[V^L_p]$ is at most $O(p^{-1} \log n)$ w.h.p. Therefore, each edge $e \in E_p$ is incident to at most $O(p^{-1}\log n)$ other edges in $E_p$ and thus every vertex in $V$ has at most $O(p^{-1}\log n)$ incident edges in $E_p$.
%
\end{proof}

\end{document}